\documentclass[12pt]{article}
\usepackage{amssymb,amsmath,graphicx,stmaryrd,mathtools, amsthm, color, ragged2e}
\usepackage{algorithm,algpseudocode}
\usepackage{enumitem,subcaption}

\usepackage{authblk}
\usepackage[left=3cm, right=3cm]{geometry}
\usepackage{hyperref}

\newtheorem{theorem}{Theorem}
\newtheorem{lemma}[theorem]{Lemma}
\newtheorem{proposition}[theorem]{Proposition}

\newtheorem{definition}[theorem]{Definition}

\newcommand{\figcaption}[1]{\caption{\small #1}}
\theoremstyle{remark}
\newtheorem{remark}{Remark}

\newcommand{\ket}[1]{\left|#1\right\rangle}
\newcommand{\norm}[1]{\left\|#1\right\|}
\newcommand{\bra}[1]{\left\langle#1\right|}

\newcommand{\wsize}{{\rm{wsize}}}
\newcommand{\leaf}{{\rm{leaf}}}
\newcommand{\cT}{\mathcal T}

\newcommand{\cS}{\mathcal S}
\newcommand{\cG}{\mathcal G}

\newcommand{\cA}{\mathcal A}

\newcommand{\cV}{\mathcal V}
\newcommand{\cM}{\mathcal M}

\newcommand{\internal}{{\rm{in}}}
\newcommand{\pseudo}{{\rm{pseudo}}}
\newcommand{\black}{{\rm{black}}}
\newcommand{\red}{{\rm{red}}}

\newcommand{\Local}{{\rm{Local}}}

\newcommand{\BlackPath}{{\rm {BlackPath}}}

\newcommand{\PostProcess}{{\rm {PostProcess}}}
\newcommand{\free}{{\rm free}}
\newcommand{\forbid}{{\rm forbid}}

\title{Time and Query-Optimal Quantum Algorithms Based on Decision Trees}
\author{Salman Beigi} 
\author{Leila Taghavi}

\author{Artin Tajdini}

\affil{\it \small QuOne Lab, Phanous Research and Innovation Centre, Tehran, Iran}

\date{}

\begin{document}

\maketitle

\begin{abstract}
It has recently been shown that starting with a classical query algorithm (decision tree) and a guessing algorithm that tries to predict the query answers, we can design a quantum algorithm with query complexity $O(\sqrt{GT})$ where $T$ is the query complexity of the classical algorithm (depth of the decision tree) and $G$ is the maximum number of wrong answers by the guessing algorithm~\cite{LL16,BT20}. 
In this paper we show that, given some constraints on the classical algorithms, this quantum algorithm can be implemented in time  $\tilde O(\sqrt{GT})$. Our algorithm is based on non-binary span programs and their efficient implementation. We conclude that various graph-theoretic problems including bipartiteness, cycle detection and topological sort can be solved in time $O(n^{3/2}\log^2 n)$ and with $O(n^{3/2})$ quantum queries. Moreover, finding a maximal matching can be solved with $O(n^{3/2})$ quantum queries in time $O(n^{3/2}\log^2 n)$, and maximum bipartite matching can be solved in time $O(n^2\log^2 n)$.
\end{abstract}

\section{Introduction}

Quantum speedup can be obtained using classical query algorithms assuming that there is another classical algorithm that guesses the answers to the queries~\cite{LL16, BT20}. 
More precisely, assume that we have a classical query algorithm $\cA$, which can be thought of as a decision tree, that computes a function $f:D_f\to [m]$, where $D_f\subseteq [\ell]^n$, and we have a \emph{guessing algorithm} $\mathcal{G}$ that predicts the answers to the queries. Assume that the query complexity of $\cA$, i.e., height of the decision tree,  is $T$ and $\cG$ makes at most $G$ mistakes in its predictions. Then, there exists a quantum query algorithm for $f$ with complexity $O(\sqrt{GT})$.
This result has been used to bound the quantum query complexity of various functions, particularly for graph-theoretic problems~\cite{LL16, BT20, KT20}.
In this paper we are interested in the question of whether such a quantum algorithm can be implemented time-efficiently (given access to QRAM).  The informal statement of our main result is as follows.

\begin{theorem}\label{thm:informal} \emph{(Informal)}
Assume that we have a classical query algorithm (decision tree) that computes a function $f:D_f\to [m]$ with $D_f\subseteq [\ell]^n$ whose query complexity (depth) is $T$.  Also, assume that
we have a guessing algorithm that tries to predict the values of queried indices, making at most $G$ mistakes.
Furthermore, assume that three subroutines which we called Local, BlackPath, and PostProcess  related to both the classical query and guessing algorithms can be implemented in time $O(\text{\rm poly}(\log n))$.\footnote{We will explain more on these subroutines in Section~\ref{sec:TimeEfficient}.}   Then there exists a quantum query algorithm computing the function $f$ with  time complexity $\tilde O \big(\sqrt{GT}\big).$\footnote{$\tilde O$ hides 
poly-logarithmic factors in $n$.}
\end{theorem}

The proof of the bound of $O\big(\sqrt{GT}\big)$ on the quantum query complexity in~\cite{BT20} is based on non-binary span programs (and the dual adversary bound). 
Span program is a powerful tool for designing quantum query algorithms~\cite{Rei09}. This tool has been generalized for functions with non-binary input alphabet~\cite{ItoJeffery15} and non-binary output~\cite{BT19}. In particular, it is known that
non-binary span programs  characterize the quantum query complexity of any function (with non-binary input  and/or non-binary output alphabets) up to a constant factor~\cite{BT19}.

Our approach to prove Theorem~\ref{thm:informal} is to implement the span-program-based algorithm of~\cite{BT20} time-efficiently.
This algorithm consists of repeated applications of two unitaries, one of which is input-dependent and the other one is the reflection through the kernel of a matrix $M$ which captures the structure of the span program (see the statement of Theorem~\ref{thm:QQAlg} for the  definition of the Matrix $M$.). Then to implement the algorithm time-efficiently, we need to implement each of these unitaries in time $O(\text{poly}\log(n))$. The former input-dependent unitary, as usual, has a simple structure and is easily implemented using two quantum queries. Then the proof of Theorem~\ref{thm:informal} reduces to the implementation of the second reflection in time $O(\text{poly}\log(n))$.

Time-efficient implementation of a span program-based quantum algorithm was first studied in~\cite{BR12} for the problem of \emph{st-connectivity} in graphs. The authors in~\cite{BR12} use the framework of quantum walks to implement the second reflection mentioned above. The point is that after a slight modification of the span program and the associated matrix $M$, the singular-value gap of $M$ around its null-space (which can also be understood as the eigenvalue gap of the normalized Laplacian of the underlying graph) becomes large. This means that, a short-time quantum walk on a bipartite graph associated with $M$ can distinguish the null-space of $M$ from its complementary subspace. Then, this quantum walk is used in~\cite{BR12} for time-efficient  implementation of the reflection through $\ker(M)$ for the st-connectivity problem.

As also mentioned in~\cite{KT20}, the first idea for time-efficient implementation of the reflection through $\ker(M)$ in our problem, is to use the quantum walk of~\cite{BR12}. To this end, we need to verify that the spectral gap of the normalized Laplacian of the starting decision tree (with certain weights on its edges) is large. However, we observe that this spectral gap is indeed very small. In fact, this spectral gap is small due to the fundamental reason that the spectral gap of the Laplacian is upper bounded by the weights of the cuts in the graph, and the weights of certain cuts in our decision tree have to remain small since they are related to a parameter of the span program called \emph{negative complexity}. Indeed, in our problem there is a trade-off between the spectral gap of the Laplacian and the complexity of the span program. This trade-off persists even after modifications of the decision tree and the span program. Therefore, we cannot adopt the approach of~\cite{BR12} for our problem and need new ideas.

\subsection{Our proof strategy} As mentioned above,
the bound on the query complexity in Theorem~\ref{thm:informal} in~\cite{BT20} in the case of binary input alphabet ($\ell =2$), is based on span programs.
Our first step to prove Theorem~\ref{thm:informal} is to generalize the construction of~\cite{BT20} and convert a decision tree along with a guessing algorithm into a non-binary span program with complexity $O(\sqrt{\log(\ell)GT})$. 
We note that this bound is an improvement over the immediate bound of $O(\log(\ell)\sqrt{GT})$ that is obtained by turning an arbitrary decision tree to a binary one while blowing its depth by a $\log(\ell)$ factor. 

Our second step is to \emph{compile} a wide class of non-binary span programs including the above ones, into a quantum algorithm with the same query complexity.\footnote{In~\cite{BT19} the fact that non-binary span programs characterize quantum query complexity is proved using the dual adversary bound, and not by directly compiling to a quantum algorithm.}
This algorithm, similar to other span program based algorithms, consists of a sequence of alternative applications of two reflections. One of these reflections is a reflection through the indices of the \emph{available input vectors}. The other one is a reflection through the kernel of a matrix $M$, which is a matrix that encapsulates the structure of the span program.

The third step of the proof is to implement the above algorithm time-efficiently. That is, to implement the aforementioned reflections time-efficiently. The first reflection is easily implementable using two quantum queries. Next, to implement reflection through $\ker(M)$, we need to take a detour. We first, by making some modifications on the starting decision tree, adjust the span program as well as the matrix $M$ so that its kernel takes a tractable structure. We \emph{explicitly compute} a basis for $\ker(M)$, and  argue that the reflection through $\ker(M)$ can be implemented in time $O(\log n)$.

Let us give some details about the structure of $\ker(M)$. Following~\cite{BT20} we think of the guessing algorithm as a coloring of edges of the decision tree by two colors: black and red. For any vertex of the decision tree, its outgoing edges correspond to outputs of the query at that vertex. We color the edge associated to the output of the guessing algorithm in black, and the other edge in red. Therefore, visiting a red edge means the guessing algorithm made a mistake. Next, the decision tree is adjusted in such a way that $\ker(M)$ consists of a vector for each edge. Vectors associated to red edges become orthogonal to all other vectors, while vectors associated to black edges have a non-trivial overlap with those of neighboring black edges. Thus, any \emph{black path} in the decision tree gives rise to a sequence of vectors in the kernel any two consecutive of which overlap. Next, we observe that by adding an extra black edge in the modified decision graph, and turning a black path to a \emph{black cycle}, we can convert those vectors to orthogonal ones via \emph{Fourier transform}. 

Theorem~\ref{thm:informal} starts with two classical algorithms, namely the classical query and the guessing algorithms. We need these classical algorithms to satisfy some properties in order to obtain Theorem~\ref{thm:informal}. One immediate property is that given a vertex of the decision tree, we need to be able to locally construct the decision tree together with the coloring of the edges around that vertex at low cost. A less trivial property imposed by the above discussion about the structure of $\ker(M)$, is that we need to be able to \emph{traverse along black paths} at low cost. This property is sometimes highly non-trivial and should be verified in applications. Moreover, the classical query algorithm requires some post-processing in order to compute the output after making all necessary queries. We need the post-processing of the classical query algorithm to have a low cost comparing to $O(\sqrt{GT})$. Finally, to implement our quantum algorithm in Theorem 1 we need access to QRAM. We give more details on these issues and assumptions in Section~\ref{sec:TimeEfficient}, where we also give the formal statement of our result.

\subsection{Applications}
Theorem~\ref{thm:informal} has several applications. For example, we can easily show that given query access to a list $x=(x_1,x_2,\ldots ,x_n)\in [\ell]^n$, the time complexity of counting all input indices with value equal to a fixed $q\in[\ell]$ in $x$ is $O(\sqrt{rn\log \ell}\log^2 n)$, where $r$ is the number of such indices. Also, given query access to $x\in\{0,1\}^n$ with the promise that its hamming weight is at most $h$, we can identify it in time $O(\sqrt{hn}\log^2 n)$.
More important applications of Theorem~\ref{thm:informal} are in proving upper bounds on quantum time complexity of various graph-theoretic problems. 

\begin{proposition}\label{pro:BFS-based}
Suppose that we have query access to the adjacency matrix of a simple\footnote{We can derive the same results for non-simple graphs by making minor modifications in the proofs.} (possibly directed) graph $\cG$ on $n$ vertices. Then, the time complexity of the following problems is $O(n^{3/2}\log^2 n)$ and their query complexity is $O(n^{3/2})$:
\begin{enumerate}
\item[\rm{(i)}] \textsc{[bipartiteness]}\label{ex:bipartiteness} Decide whether $\cG$ is bipartite or not.
\item[\rm{(ii)}] \textsc{[cycle detection]}\label{ex:cycle} Decide whether $\cG$ is a forest or has a cycle.

\item[\rm{(iii)}] \textsc{[directed st-connectivity]}\label{ex:dir-st-con} Find a shortest path (the path that consists of the least number of edges) between two vertices $s$ and $t$ in a directed or undirected graph $\cG$.  

\item[\rm{(iv)}]  \textsc{[smallest cycles containing a vertex]} 
Find the length of the smallest \emph{directed} cycle containing a given vertex $v$ in a directed graph $\cG$.
\end{enumerate}
Also the following holds:
\begin{enumerate}
\item[\rm{(v)}]  \textsc{[$k$-cycle containing a vertex]}  There is a quantum algorithm for deciding whether $\cG$ has a cycle of length $k$, for a fixed $k$, containing a given vertex $v$ that is run in time $O((2k)^{(k-1)}n^{3/2}\log^2 n)$ and makes $O((2k)^{(k-1)}n^{3/2})$ queries.

\end{enumerate}
\end{proposition}

The bounds on the quantum query complexities in this proposition have been proved in Proposition 9 of~\cite{BT20}. To prove this proposition using  Theorem~\ref{thm:informal}, it is enough to show that the classical time complexities of the aforementioned  operations on the associated decision tree (local construction of the decision tree, and traversing the black paths) are at most $O(\log n)$. We note that the classical algorithm behind all of the above results is  the BFS algorithm, for which we show that the costs of the above operations are $O(\log n)$. 
We give more details on this in Section~\ref{sec:applications-bfs}.

\begin{proposition}\label{pro:DFS-based}
Suppose that we have query access to the adjacency matrix of a directed graph $\cG$ on $n$ vertices. Then the time complexity of the following problems is $O(n^{3/2}\log^2 n)$ and their query complexity is $O(n^{3/2})$:
\begin{enumerate}
\item[\rm{(i)}] \textsc{[topological sort]}\label{ex:topsort} Assuming that $\cG$ is acyclic, find a vertex ordering of $\cG$ such that for all directed edges $(u,v)$, $u$ appears before $v$.

\item[\rm{(ii)}] \textsc{[connected components]}\label{ex:con-com} 
Find all connected components of $\cG$.

\item[\rm{(iii)}] \textsc{[strongly connected components]}\label{ex:str-con-com}
Find strongly connected components of $\cG$.
Recall that two vertices $u, v$ belong to the same strongly connected component iff there exists a directed path from $u$ to $v$ and a directed path from $v$ to $u$ in $\cG$.
\end{enumerate}
\end{proposition}

Again, the bounds on the quantum query complexities in this proposition have been proved in~\cite{BT20}.
Proof of this proposition is very similar to that of Proposition~\ref{pro:BFS-based}. The main difference is that the classical algorithm for these problems in~\cite{BT20} are based on the DFS algorithm.  The DFS algorithm is similar to the BFS algorithm; the only difference is that instead of a queue that is a first-in first-out list, DFS algorithm uses a stack which is a last-in first-out list. Thus, its time-efficient implementation is feasible using the same ideas that we use for the BFS algorithm.

\begin{table}[]
\begin{center}
{\small
\begin{tabular}{|l|l|l|l|} 
\hline
Problem
                      &  matrix model &  list model \\ \hline
\textsc{bipartiteness} & $O(n^{3/2}\log^2 n)^*$ & $O\big(\sqrt{(m+n)n}\log^{5/2}(n)\big)$ \\
\textsc{topological sort}  & $O(n^{3/2} \log^2 n)$ & $O\big(\sqrt{(m+n)n}\log^{5/2}(n)\big)$ \\
\textsc{directed st-connectivity} & $O(n^{3/2} \log^2 n)$ & $O\big(\sqrt{(m+n)n}\log^{5/2}(n)\big)$ \\ 
\textsc{connected components}   & $O(n^{3/2} \log^2 n)$ & $O\big(\sqrt{(m+n)n}\log^{5/2}(n)\big)$ \\
\textsc{strongly connected components} & $O(n^{3/2}\log^2 n)$ & $-$ \\
\textsc{cycle detection} & $O(n^{3/2}\log^2 n)^*$ & $-$\\
\textsc{SCCV} & $O(n^{3/2}\log^2 n)$ & $-$\\
\textsc{$k$-cycle containing a vertex} &  $O((2k)^{(k - 1)} n^{3/2} \log^2 n)$ &  $-$\\
\textsc{maximum bipartite matching} & $O(n^2 \log^2 n)^*$ & $O\big(n\sqrt{m+n}\log^{5/2} n\big)^*$ \\
\textsc{maximal matching} & $O(n^{3/2} \log^2 n)^*$ & $O\big(\sqrt{(m+n)n}\log^{5/2} n\big)^*$ \\
\hline
\end{tabular}
}
\end{center}
\caption{\it \small Summary of our results on the \emph{time} complexity of graph-theoretic problems. Here, \textsc{SCCV} stands for Smallest Cycles Containing a Vertex. We are giving algorithms for \textsc{SCCV} and $k$-\textsc{cycle containing a vertex} here for the first time. Moreover, our algorithm for \textsc{Directed st-connectivity} finds an st-path while the algorithm of~\cite{BR12} only \emph{detects} the presence of an st-path. Bounds labeled by $*$ have been studied in previous works, but their exact poly-logarithmic factors are not explicitly given, so we cannot compare their results with ours (see Subsection~\ref{ssec:related} for more details). 
}
\label{table}
\end{table}

\begin{proposition}\label{pro:BFS-list-based} 
Suppose that we have query access to the adjacency list of graph $\cG$ with $n$ vertices and $m$ edges. Then the time complexity of the following problems is $O\big(\log^{5/2}(n)\sqrt{(m+n)n}\big)$.
\begin{enumerate}
\item[\rm{(i)}] \textsc{[directed st-connectivity]}\label{ex:dir-st-con-list} Find a shortest (directed or undirected) path between two vertices $s, t$ in $\cG$.

\item[\rm{(ii)}] \textsc{[bipartiteness]}\label{ex:bipartiteness-list} Decide whether $\cG$ is bipartite or not.

\item[\rm{(iii)}] \textsc{[topological sort]}\label{ex:topsort-list} Assuming that $\cG$ is acyclic, find a vertex ordering of $\cG$ such that for all $(u,v)\in E$, $u$ appears before $v$.

\item[\rm{(iv)}] \textsc{[connected components]}\label{ex:con-com-list} 
Determine all connected components of $\cG$.
 
\end{enumerate}
\end{proposition} 

The proof of this proposition, again based on the quantum query algorithms of~\cite{BT20}, is similar to those of Proposition~\ref{pro:BFS-based} and Proposition~\ref{pro:DFS-based}.

\begin{proposition}[\textsc{maximum bipartite matching}]\label{pro:bipartite} Assume that $\cG$ is unweighted and bipartite. Then there is a quantum algorithm for finding a maximum bipartite matching in $\cG$ that runs in time $O(n^2\log^2 n)$ and makes $O(n^2)$ queries in the adjacency matrix model. For the adjacency list model the time complexity is $O\big(n\sqrt{m+n}\log^{5/2} n\big)$.
\end{proposition}

The proof of this proposition is based on the classical algorithm of Hopcroft and Karp~\cite{HK73}. Nevertheless, we note that the bounds on the quantum query complexity of finding the maximum bipartite matching is $O(n^{7/4})$ in the adjacency matrix model~\cite{LL16} and $O(n^{3/4}\sqrt{m+n})$ in the adjacency list model~\cite{BT20}. However, here, in order to be able to implement the \emph{traverse of the black paths} mentioned above at low cost, we need to increase the number of queries to $\Omega(n^2)$.

\begin{proposition}[\textsc{maximal matching}]\label{pro:maximal}
The time complexity of finding a maximal matching in an unweighted graph $\cG$ is $O(n^{3/2}\log^2 n)$ in the adjacency matrix model, and $O(\sqrt{(m+n)n}\log^{5/2} n)$ in the adjacency list model.
\end{proposition}
 The classical algorithm that we use to prove this proposition, extends a matching in a greedy manner.

\subsection{Related works} \label{ssec:related}
Assume that we have a time-efficient quantum algorithm computing a function $f$. It has been recently shown that we can convert this algorithm into a span program and compile it back into a quantum algorithm computing the same function, with a poly-logarithmic overhead in its time complexity~\cite{CJOP20}. This suggests that span programs not only characterize the quantum query complexity of a function, but also its quantum time complexity up to a poly-logarithmic factor. 

The problem of st-connectivity on a graph $G$ given in adjacency matrix model has a quantum algorithm with time $\tilde{O}(\sqrt{kn})$, where $k$ is the distance from $s$ to $t$~\cite{BR12}. This algorithm only detects the presence of an st-path. Our algorithms for this problem in Proposition~\ref{pro:BFS-based} and Proposition~\ref{pro:BFS-list-based}, however, \emph{find} an st-path poly-logarithmically faster.

The problems of detecting bipartiteness, and whether $G$ is a forest or contains a cycle, have space-efficient quantum query algorithms with time complexity $\tilde{O}(n^{3/2})$ in the adjacency matrix model and $\tilde{O}(n\sqrt d)$ in the adjacency list model, where $d$ is the maximum degree of $G$~\cite{CMB16}. 

The time and query complexities of some of graph-theoretic problems including minimum spanning tree, connectivity, strong connectivity and single source shortest paths, are also studied in~\cite{DHHM06}.

The time complexity of bipartite matching is investigated by Ambainis and {\v{S}}palek~\cite{AS06}. Also, the problem of maximal matching is studied by D\"{o}rn~\cite{Dorn09}.


We note that we can give a proof of Theorem~\ref{thm:informal} using the algorithm of Lin and Lin~\cite{LL16}. Their algorithm is based on finding the first mistake of the guessing algorithm at each stage; starting from the root of the decision tree find the first mistake, then look for the next mistake etc. until reaching a leaf.
To implement this algorithm we need to use the BlackPath subroutine mentioned in the statement of Theorem~\ref{thm:informal}. Nevertheless, repeating the ``find-the-first-mistake subroutine" many times may introduce errors. To reduce this error we need to repeat the whole algorithm logarithmically many times. This introduces an extra logarithmic factor in the query and time complexity of the algorithm, giving a worse bound comparing to what we prove with our method (our precise bound is given in Theorem~\ref{thm:binaryClassical2quantumtime}). 

\paragraph{Structure of the paper:} In the following section we fix some notations and define non-binary span programs. In Section~\ref{sec:decision tree-alg}  we show that how given a pair of decision tree and a guessing algorithm, we can turn them to a non-binary span program, and also describe the general structure of the quantum algorithm associated to that span program.
Our main contribution, namely the proof of Theorem~\ref{thm:informal} comes in Section~\ref{sec:TimeEfficient}. Proofs of the above propositions on the applications of Theorem~\ref{thm:informal} come in Section~\ref{sec:applications-bfs}.

\section{Preliminaries}
To compute a function $f:D_f\to[m]$, with domain $D_f\subseteq [\ell]^n$, in the query model, it is assumed that access to the input $x= (x_1, \dots, x_n)\in D_f$ of $f$ is provided by queries to its coordinates. 
In a classical query algorithm, we ask the value of some coordinate and based on its answer, decide to query another coordinate; at the end we output the result. 

Such a classical query algorithm can be modeled by a \emph{decision tree} $\cT$ whose internal vertices are associated with queries, i.e., indices $1\leq j\leq n$, and whose edges correspond to answers to queries, i.e., elements of $[\ell]$. At each vertex the algorithm queries the associated index, and moves to the next vertex via the edge whose label equals the answer to that query. The algorithm ends once we reach a leaf of the tree. The leaves are labeled by elements of $[m]$, the output set of the function, and determine the output of the algorithm. The query complexity of the algorithm is the maximum number of queries in the algorithm over all $x\in D_f$, which is equal to the depth of the decision tree.

When the input alphabet is non-binary ($\ell>2$), it is sometimes useful to label outgoing edges of a vertex of the decision tree $\cT$ not by elements of $[\ell]$, but with subsets of $[\ell]$ that form a \emph{partition}. The point is that sometimes the query algorithm is ignorant of the exact value of a query answers, and depends only on a subset to which the query answer belongs. We refer to~\cite{BT20} for more details and explicit examples on this.

In the quantum case, queries can be made in superposition. Such a query to an input $x$ can be modeled by the unitary  $O_x$:
\begin{equation*}
O_x|j,p\rangle=|j,(x_j+p) \mod \ell \rangle,
\end{equation*}
where the first register contains the query index $1\leq j\leq n$, and the second register stores the value of $x_j$ in a reversible manner. Therefore, a quantum query algorithm for computing $f(x)$ starts with a computational basis state (that is input independent), applies a unitary of the form 
\begin{equation*}
U_kO_x\ldots U_2O_xU_1,
\end{equation*}
where $U_i$'s are input independent (yet they depend on $f$ itself), and ends
by a measurement in the computational basis that determines the outcome of the algorithm. We say that an algorithm computes $f$, if for every $x\in D_f\subseteq [\ell]^n$ the algorithm outputs  $f(x)$ with probability at least $2/3$.

The query complexity of such an algorithm is the number of queries, i.e., the number of $O_x$'s in the sequence of unitaries. 
The time complexity of this algorithm equals the sum of the time complexity of implementing individual unitaries where the time complexity of implementing $O_x$ is assumed to be $1$.

\subsection{Non-binary span program}\label{sec:NBSP}
A \emph{non-binary span program} (NBSP) evaluating a function $f:D_f\rightarrow[m]$ with $D_f\subseteq [\ell]^n$ consists of:
\begin{itemize}
\item a finite-dimensional inner product space of the form 
$$H=H_1\oplus H_2\oplus\ldots H_n\oplus H_{\rm free}\oplus H_{\rm forbid},$$ 
where each $H_j$ for $1\leq j\leq n$ can be written as 
$H_j=H_{j,0}+\cdots+H_{j,\ell-1}$. Also, $H_\free$/ $H_{\rm forbid}$ is the space of the vectors that are \emph{always available/unavailable}. These two extra vector spaces are sometimes useful in the implementation of the algorithm associated to a span program.
\item a finite-dimensional inner product  vector space $\cV$,
\item $m$ non-zero target vectors $|t_0 \rangle, |t_1 \rangle,\ldots ,|t_{m-1} \rangle\in \cV$,
\item a linear operator $A: H\to \cV$.
\end{itemize}
Given these data, for any $x\in D_f$, let 
$$H(x):=\bigoplus_{1 \leq j \leq n}H_{j,x_j}\oplus H_{\rm free}\subseteq H.$$ 
Then $AH(x)$, i.e., the image of $H(x)$ under $A$, is called the \emph{space of available vectors} for $x$.
We say that the span program evaluates $f$ if for every $x\in D_f$ there exists 
\begin{itemize}
\item a positive witness $\ket{w_x}\in H(x)$ such that $A\ket{w_x}=\ket{t_{f(x)}}$ and
\item a negative witness $\ket{\bar{w}_x}\in \cV$ such that $\bra{\bar{w}_x}AH(x)=0$ and $\bra{\bar{w}_x}t_\beta\rangle=1-\delta_{f(x),\beta}$ for all $\beta\in[m]$.
\end{itemize}
Here, the first condition guarantees that $\ket{t_{f(x)}}$ belongs to the space of available vectors, and the second condition says that $\ket{t_\beta}$ for $\beta\neq f(x)$ does not belong to the space of available vectors.

We say that the span program evaluates the function $f$ if for every $x\in D_f$, $\ket{t_\alpha}$ belongs to the span of the available vectors $I(x)$ if and only if $\alpha=f(x)$. 

The positive and negative complexities of the NBSP together with the collections $w$ and $\bar w$ of positive and negative witnesses are defined by
\begin{align*}
&\mathrm{wsize}^+(w,\bar{w}):=\max_{x\in D_f} ~\|\ket{w_x}\| ^2,\\
&\mathrm{wsize}^-(w,\bar{w}):=\max_{x\in D_f} ~\|A^\dagger \ket{\bar{w}_x}\| ^2.
\end{align*}
Then the complexity of the NBSP is equal to
\begin{align}\label{eq:wsize-pwwb}
 \mathrm{wsize}(P,w,\bar{w})=\sqrt{\mathrm{wsize}^-(P,w,\bar{w})\,\cdot\,\mathrm{wsize}^+(P,w,\bar{w})}. 
\end{align}

\paragraph{Non-binary span program with orthogonal inputs:} We use the non-binary span program of~\cite{BT19} in the \emph{orthogonal} sense, that are particularly useful when $\ell=2$. These are a restricted class of NBSPs, so here we mention their differences. Suppose that for each $j$, there is an \emph{orthonormal} basis for $H_j$ such that each of its subspaces $H_{j, q}$ is spanned by a subset of that basis. 
Then, we may denote the image of each of these subsets under the map $A$ by $I_{i,q}\subseteq \cV$. We note that, in this case,
the subspaces $H_{j, q}$ are determined by the index sets of those orthonormal bases, which are the same sets that index vectors in $I_{j, q}$'s. Thus, to describe
such a span program, we may only specify subsets $I_{j,q}\subseteq \cV$ for every $1\leq j\leq n$ and $q\in [\ell]$, as well as $I_\free, I_\forbid\subset \cV$.
Then, the set of input vectors $I\subseteq \cV$ is defined by 
$$I=I_\free\cup I_\forbid\cup \Bigg(\bigcup_{j=1}^n \bigcup_{q\in [\ell]} I_{j,q}\Big),$$
and for every $x\in D_f$ the set of \emph{available vectors} $I(x)$ is defined by 
$$I(x)=I_\free\cup \bigcup_{j=1}^n I_{j,x_j}.$$
Target vectors, and negative/positive witnesses remain the same. 
Note that in this case, the linear operator $A$ can be thought of as a $d\times |I|$ matrix consisting of all input vectors as its columns where $d=\dim \cV$.

\section{From decision trees to quantum algorithms}\label{sec:decision tree-alg}
In this section we first show how we can convert a
non-binary span program with a special structure to a quantum query algorithm. We then prove how we can convert a classical query algorithm that computes a function $f$ and a guessing algorithm that predicts queries to its input, into a non-binary span program. This non-binary span program has the required special structure.

\subsection{From the span program to a quantum algorithm}

In order to use span programs for proving Theorem~\ref{thm:informal} we need to describe the quantum algorithm associated to a non-binary span program. 

\begin{theorem}\label{thm:QQAlg}
Consider a non-binary span program for a function $f:D_f\to [m]$ with $D_f\subseteq [\ell]^n$.
Suppose that the target vectors of this span program are of the form
$$\ket{t_\alpha} = \ket{z_0} - \ket{z_{\alpha}}, \qquad 1\leq \alpha\leq  m,$$
for some vectors $\ket{z_1}, \dots, \ket{z_m}\in \cV$, and the negative witness $\ket{\bar{w}_x}$ for every $x$ satisfies
$$\bra{\bar{w}_x}z_0\rangle=1, \quad \text{ and } \quad \bra{\bar{w}_x} z_\alpha\rangle=\delta_{f(x),\alpha}, \qquad 1\leq \alpha\leq m.$$
Let $\epsilon>0$ be a constant, and $\wsize^+, \wsize^-$ be the positive and negative witness sizes of the span program.
Let $A'$ be a matrix with columns $\frac{\sqrt 2\epsilon}{\sqrt{\wsize^+}}\ket{z_0} , \frac{-\sqrt 2\epsilon}{\sqrt{\wsize^+}}\ket{z_1} , \ldots , \frac{-\sqrt 2\epsilon}{\sqrt{\wsize^+}}\ket{z_m}$, i.e.,
\begin{align*}
    A'   = &
     \frac{\sqrt 2\epsilon}{\sqrt{\wsize^+}} 
    \begin{bmatrix}
    \ket{z_0} & -\ket{z_1} & \ldots & -\ket{z_m}
    \end{bmatrix} \\ = &
    \frac{\sqrt 2\epsilon}{\sqrt{\wsize^+}} \Big(\ket{z_0}\bra{e_{z_0}}-\sum_{\alpha=1}^m\ket{z_\alpha}\bra{e_{z_\alpha}} \Big)
    ,
\end{align*}
where $e_{z_j}$'s are indices of the columns of $A'$ .
Also, let $M=[A', A]$ be a column-padded matrix, with first the columns of $A'$ and then the columns of $A$.  Equivalently, define
$M:{\rm span}\{\ket{e_{z_0}}, \ket{e_{z_1}}, \dots, \ket{e_{z_m}}\}\oplus H\to \cV$  by $M=A'\oplus A$.
Let  $\Lambda$ be the orthogonal projection on the kernel of $M$, and let $$\Pi_x=\sum_{j\in\{0,\ldots,m\} } \ket{j}\bra{j}+ \Pi'_{x},$$
where $\Pi'_{x}$ denotes projection on $AH(x)$.
Then, there is a bounded error quantum algorithm computing $f(x)$ based on repeated applications of the reflections $2\Lambda-I$ and $2\Pi_x-I$ for $O(W)$ times, where $W=\sqrt{\wsize^+\wsize^-}$ is the complexity of the span program.

\end{theorem}

The proof of this theorem given in Appendix~\ref{app:QQAlg}, is based on ideas from~\cite{LMRSS11}.

\medskip

In the statement of Theorem~\ref{thm:QQAlg} we restrict the structure of the target vectors and negative witnesses. We note that any non-binary span program can be transformed to a \emph{canonical form}~\cite{BT19} whose target vectors and negative witnesses do have this structure. Thus, this theorem can be applied on any span program after transforming it to a canonical one.

\subsection{From a decision tree to a span program}

In order to use Theorem~\ref{thm:QQAlg} towards the proof of Theorem~\ref{thm:informal}, we need to first construct a non-binary span program based on the setup of Theorem~\ref{thm:informal}. 

\begin{theorem}\label{thm:Classical2quantum}
Assume that we have a classical algorithm  for a function $f:D_f\to [m]$ with $D_f\subseteq [\ell]^n$ whose query complexity is $T$.  Furthermore, assume that we have a guessing algorithm that predicts the values of queried bits, making at most $G$ mistakes. Then there exists a non-binary span program  computing the function $f$ with query complexity $O(\sqrt{GT})$.
\end{theorem}

We note that this result in the special case of $\ell=2$ is already proven in~\cite{BT20}. Furthermore, the existence of such a non-binary span program in the general case can already be proven using the known bound of $O(\sqrt{GT})$ on the quantum query complexity, and the fact that NBSPs characterize quantum query complexity up to a constant factor. Nevertheless, here in order to use Theorem~\ref{thm:QQAlg} to prove Theorem~\ref{thm:informal}, we need the explicit structure of the span program given by Theorem~\ref{thm:Classical2quantum}.

Let $f:D_f\to [m]$ with $D_f\subseteq [\ell]^n$ be an arbitrary function. Also, let $\cT$ be a decision tree for $f$ with depth $T$. This means that
internal vertices of $\cT$ are indexed by $j\in \{1, \dots, n\}$, and outgoing edges are labeled by elements (or in general subsets in partitions) of $[\ell]$. Moreover, leaves are indexed by elements of $[m]$. Following~\cite{BT20}, we use two colors, black and red, to present the behavior of the guessing algorithm. At each vertex with associated index $j\in \{1, \dots, n\}$, the guessing algorithm predicts a value for $x_j$. We color the outgoing edge of that vertex whose label equals (contains) the predicted value in black. We color the rest of the outgoing edges in red. We call such a coloring of edges of a decision tree a \emph{G-coloring}. Note that, in a G-coloring, each vertex has exactly one black outgoing edge. Also note that, by the assumption in the theorem, in each path from the root to a leaf of $\cT$, there are at most $G$ red edges. 

Here, for simplicity of presentation and the fact that to prove Theorem~\ref{thm:informal}, the non-binary span program with binary inputs is sufficient, we prove Theorem~\ref{thm:Classical2quantum} in the special case of $\ell=2$.   
This proof, taken from~\cite{BT20} is based on a span program with orthogonal inputs. The proof of this theorem in the general case is given in Appendix~\ref{app:nonbinary}.

\begin{proof}
To prove the theorem we present a non-binary span program with complexity $O(\sqrt{GT})$. Our span program however, is not for $f$ but for the function $\tilde f$ which sends any $x\in D_{f}$ to an associated leaf of the decision tree $\cT$. Observe that in the classical algorithm we start with the root and follow edges of $\cT$ labeled by $x_j$'s until we reach a leaf which we denote by $\tilde f(x)$. As the classical algorithm essentially finds $\tilde f(x)$ (from which $f(x)$ can be obtained), to prove the theorem we just need to design a span program for $\tilde f$.

To present this span program first we need to develop some notation. Let $V=V_\internal\cup V_{\rm leaf} $ be the vertex set of $\mathcal T$, where $V_\internal$ denotes the set of internal vertices, and $V_{\rm leaf}$ the leaves of $\cT$. Then, for every vertex $v\in V_\internal$, its associated index is denoted by $J(v)$, i.e., $J(v)$ is the index $1\leq j\leq n$ that is queried by the classical algorithm at node $v$. The two outgoing edges of $v$ are indexed by elements of $\{0,1\}$ and connect $v$ to two other vertices. We denote these vertices by $N(v, 0)$ and $N(v, 1)$. We also represent the G-coloring of edges of $\cT$ by a function $C(v, q)\in \{\black,\red\}$ which denotes the color of the outgoing edge of $v$ with label $q$.

We can now describe our non-binary span program:
\begin{itemize}
\item the vector space $\cV$ is $|V|$-dimensional with the orthonormal basis 
$\{\ket{v}\,:\, v\in V\},$

\item the input vectors are
$$I_{j,q}=\Big\{\sqrt{W_{C(v, q)}}\big(\ket{v}-\ket{{N(v, q)}}\big)  \,:\, v\in V \text{ s.t. } J(v)=j \Big\},$$
where $W_{\black}$ and $W_{\red}$ are positive real numbers to be determined, 
\item the target vectors are indexed by leaves $z$ of the tree: 
$$\ket{t_z}=\ket{z_0}-\ket{z},$$
where as before $z_0\in V$ is the root of $\cT$.
\end{itemize}

For every $x\in D_f$ let $P_x=P_{\tilde f(x)}$ be the path from the root of the decision tree to the leaf $\tilde f(x)$. Thus the target vector $\big|t_{\tilde f(x)}\big\rangle$ equals
$$\big|t_{\tilde f(x)}\big\rangle= \ket{z_0} - \big|\tilde f(x)\big\rangle=\sum_{v\in P_{x}} \frac{1}{\sqrt{W_{C\big(v, x_{J(v)}\big)}}} \left\{\sqrt{ W_{C\big(v, x_{J(v)}\big)}} \left(\ket{v}-\ket{N(v,x_{J(v)})}\right)\right\},$$
where the vectors in the braces are all available for $x$. Since by assumption the number of red edges along the path $P_{x}$ is at most $G$ and the number of all edges of $P_x$ is at most $T$, the positive complexity is bounded by
$$\wsize^+\leq \frac{1}{W_{\red}}G+\frac{1}{W_{\black}}T.$$
We let the negative witness for $x$ be 
$$\ket{\bar{w}_x}=\sum_{v \in P_{x} } \ket{v}.$$
It is easy to verify that $\ket{\bar w_x}$ is orthogonal to all available vectors, and that $\bra{\bar w_x} t_u\rangle = \bra{\bar w_x} z_0\rangle =1$ for all $u\neq \tilde f(x)$. Thus $\ket{\bar w_x}$ is a valid negative witness. Moreover, an input vector of the form
$$\sqrt{W_{C(v, q)}}\big(\ket{v}-\ket{{N(v, q)}}\big),$$
contributes in the negative witness size only if its corresponding edge $(v, N(v, q))$ leaves the path $P_x$, i.e., they have only the vertex $v$ in common. In this case the contribution would be equal to $W_{C(v, q)}$, the weight of that edge. The number of such red (black) edges equals the number of black (red) edges in $P_x$, which is bounded by $ T$ ($G$). Therefore, the negative witness size is bouned by
$$\wsize^-\leq  W_{\black} G+ W_{\red} T\big.$$
Now letting $W_\black=\frac{1}{G}$ and $W_\red=\frac{1}{T}$, we have $\wsize^+\leq 2GT$ and $\wsize^-\leq 2$.
Therefore, the quantum query complexity of $\tilde f$, and hence also $f$ are bounded by $O(\sqrt{GT})$. 
\end{proof}

\section{Time-efficient implementation}\label{sec:TimeEfficient}
This section is dedicated to the proof of our main result that is formally stated in Theorem~\ref{thm:binaryClassical2quantumtime} below. To prove our result, we start with the query-efficient algorithm provided by 
Theorem~\ref{thm:QQAlg} and Theorem~\ref{thm:Classical2quantum} and try to implement it time-efficiently. 

In order to obtain a time-efficient quantum algorithm, in addition to the query complexity of the underlying classical algorithms that we start with, their time complexity also matter. Indeed, we need the cost of certain  operations in the classical algorithms that build the decision tree and the appropriate coloring of its edges to be bounded. 

The state space of our proposed quantum algorithm described by Theorem~\ref{thm:QQAlg} is essentially the edge set (or the vertex set) of the decision tree $\cT$. Indeed, in the classical algorithm also each vertex of the decision tree corresponds to a \emph{state} of the algorithm. Thus the state space of our quantum algorithm is essentially the same as the state space of the underlying classical algorithm. However, in order to implement the algorithm time-efficiently we need to keep track of the correspondence between the state space of the algorithm and vertices of the decision tree. This correspondence is provided by some local operations mentioned above. Thus, in the following we assume that the following classical subroutines have bounded time complexity:

\begin{itemize}

\item \textbf{Local subroutine:} We assume that there is a subroutine $\cA_{\Local}$ that given a state of the classical algorithm (a vertex of the decision tree) outputs the local structure of the decision tree in a neighborhood of that vertex. In particular given a vertex, $\cA_\Local$ returns the followings:
\begin{itemize}
    \item The type of the vertex i.e. root, leaf, or internal vertex
    \item The parent vertex, list of children vertices and the color of the edges connecting them
    \item The query associated to the vertex and the query answer associated with each of the neighboring edges
\end{itemize}

\item \textbf{BlackPath subroutine:} Note that any vertex of the decision tree, except leaves connected to a red edge, belongs to a path of black edges; see Figures~\ref{fig:OR} for an illustrative example. We assume that there is a subroutine $\cA_{\BlackPath}$ that given a state of the classical algorithm (a vertex of the decision tree), returns the followings about the black path containing that vertex:
\begin{itemize}
    \item The \emph{length} of the black path.
    \item Given additional input $k$, the $k$-th vertex on the black path.  
\end{itemize}

\item \textbf{PostProcess subroutine:} The classical algorithm after making all the queries and moving from the root of the decision tree to a leaf, makes some (classical) post-processing in order to compute the final output of the function.  We assume that this is done by a subroutine denoted by $\cA_{\PostProcess}$.

\end{itemize} 

As discussed above the necessity of subroutines $\cA_{\Local}, \cA_{\PostProcess}$ is apparent. However, the necessity of $\cA_{\BlackPath}$ will become clear later. 

We note that $\cA_{\BlackPath}$, gives global information about the structure of the decision tree beyond local information. Nevertheless, recall that black edges correspond to making correct guesses of query answers. Then, we expect that updating the state of the algorithm as we move along a black path is not costly. 

Finally, before giving the formal statement of our main result, we need to clarify a few points regarding our setup for quantum algorithms: 
\begin{itemize}
    \item[-] In our quantum algorithms we assume that \emph{all one-qubit gates} as well as the CNOT gate are available. We also assume the full connectivity of the qubits, meaning that CNOT can be applied on any pair of qubits.
    \item[-] The time (or circuit) complexity of algorithms are measured in terms of the number of gates and the number of queries.
    \item[-] More importantly, we assume accessibility to quantum RAM (QRAM).
\end{itemize}


\begin{theorem}\label{thm:binaryClassical2quantumtime}
Assume that we have a classical algorithm that computes a function $f:D_f\to [m]$ with $D_f\subseteq [\ell]^n$ whose query complexity is $T$.  Furthermore, assume that we have
\begin{itemize}
\item a guessing algorithm that predicts the values of queried bits that makes at most $G$ mistakes, 
\item classical subroutines $\cA_\Local, \cA_\BlackPath$ and $\cA_{\PostProcess}$ described above that can be implemented in time $C_\Local, C_\BlackPath$ and $C_{\PostProcess}$ respectively.
\end{itemize}
  Then, having access to QRAM, there is a quantum algorithm computing the function $f$ with query complexity $O(\sqrt{\log(\ell)GT+ \log^2(\ell)G^2})$ and time (circuit)  complexity 
  $$O \left(\big(C_\Local+C_\BlackPath+\log^2 n\big)\sqrt{\log{(\ell)}GT+\log^2(\ell)G^2} + C_{\PostProcess}\right).$$
In particular, if $\log(\ell)G=O(T)$, that is usually the case in applications, the query complexity equals $O(\sqrt{\log(\ell)GT})$ and the. time complexity is
  $$O \left(\big(C_\Local+C_\BlackPath+\log^2 n\big)\sqrt{\log{(\ell)}GT} + C_{\PostProcess}\right).$$ 
\end{theorem}

Note that a decision tree may have an  exponential number of edges, so the state space of the classical algorithm may have at least linear size. That is, to encode the state of the algorithm we may need at least $\Omega(n)$ bits. If this encoding is not done properly, to update the state of the algorithm in each step, we need to update all these $\Omega(n)$ bits, and the cost of implementing subroutines $\cA_\Local$ and $\cA_\BlackPath$ becomes at least $\Omega(n)$. To reduce this cost, we need to encode the state of the algorithm deliberately in such a way that we do not need to read the entire encoding of a state in each update.

Our proof of Theorem~\ref{thm:binaryClassical2quantumtime} is based on the algorithm provided by Theorem~\ref{thm:QQAlg}.
This algorithm consists of repeated  applications of the unitary $U=R_\Lambda R_{\Pi}$, where $R_\Pi=2\Pi-I$ and $R_\Lambda=2\Lambda-I$ are specified by Theorem~\ref{thm:Classical2quantum}. Thus to prove our result we need to implement both $R_{\Pi}$ and $R_\Lambda$ in time $O(C_\Local+C_\BlackPath+\log^2 n)$.
Implementation of $R_\Pi$ is straightforward and will be taken care of later.  However, time-efficient implementation of $R_\Lambda$ is more involved. 

Recall that $R_{\Lambda}$ is the reflection through the kernel of $M$. Thus the implementation of $R_{\Lambda}$ amounts to the inspection of the structure of $\ker( M)$. To this end, in the following we start with the easy case of $\ell=2$ and $G=1$. We observe that by slightly changing the underlying span program (while keeping its complexity of the same order), $\ker( M)$ takes a nice structure. Then employing the Fourier transform, we would be able to implement $R_{\Lambda}$ in time $O(C_\Local +C_\BlackPath+\log^2 n)$. 

\subsection{Special case of $\ell=2$ and $G=1$} \label{ssec:G=1}
Let $f$ be the function that outputs the index of the first marked element in a list. More precisely, let $f:\{0,1\}^n\to \{0,\dots, n\}$ be the function with $f(x)$ being the smallest index $j$ with $x_j=1$. If no such index exists, then $f(x)=0$. Construction of a decision tree for this function is easy: at step $j$ we query $x_j$ and guess that the outcome is $0$. The height of this tree is $T=n$, and the first incorrect guess reveals $f(x)$ meaning that $G=1$. This decision tree, with an extra vertex $z_0$ and an extra edge $r_0$ whose roles will become clear later, is depicted in Figure~\ref{fig:sub-OrDtree}. We conclude that the query complexity of $f$ 
is $O(\sqrt n)$.
The construction of the associated span program to this decision tree and the associated matrix $M$ is as before. The only difference is that here we add an extra vertex $z_0$ as the root of the tree and connected it with edge $r_0$ to the previous root. We assume that $z_0$ corresponds to no query, so the vector $\sqrt{W_\red}(\ket{z_0}-\ket{v_0})$ is always available (it is a free vector). This edge in Figure~\ref{fig:sub-OrDtree} is depicted in red since its weight in $\sqrt{W_\red}(\ket{z_0}-\ket{v_0})$ is $W_{\red}$.  Thus the vertex set of $\cT$ is
\begin{equation*}
 V=\{z_0,z_1,\ldots,z_{n+1}\}\cup \{v_0, \dots, v_{n-1}\},
\end{equation*}
where  $z_0$ is the root, $z_i$'s are leaves and $v_i$'s are internal nodes. The edge set of $\cT$ is
\begin{equation*}
E= \{b_1,\ldots b_n\} \cup\{r_0,\ldots, r_n\},
\end{equation*} 
where $b_i$'s are black and $r_i$'s are red edges. The label of edges $r_1, \dots, r_n$ is $1$, and the label of edges $b_1, \dots, b_n$ is $0$.

\begin{figure}
\begin{subfigure}{.5\textwidth}
  \centering
  \includegraphics[scale=.95]{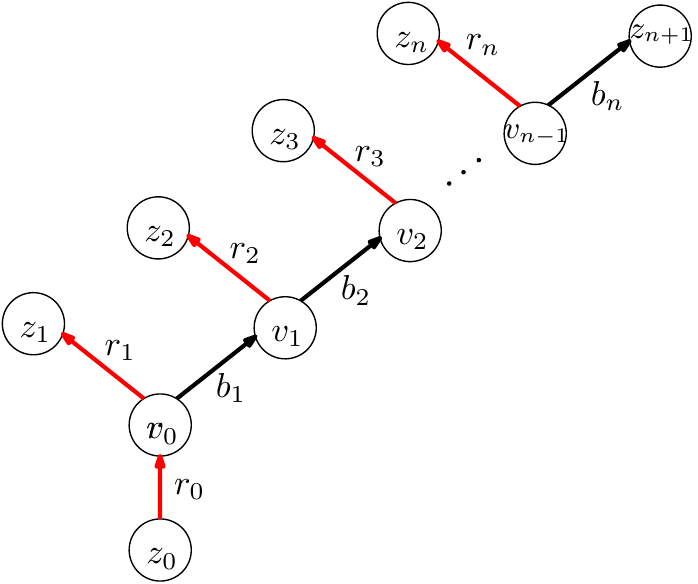}  
  \caption{}
  \label{fig:sub-OrDtree}
\end{subfigure}
\begin{subfigure}{.5\textwidth}
  \centering
  \includegraphics[scale=.95]{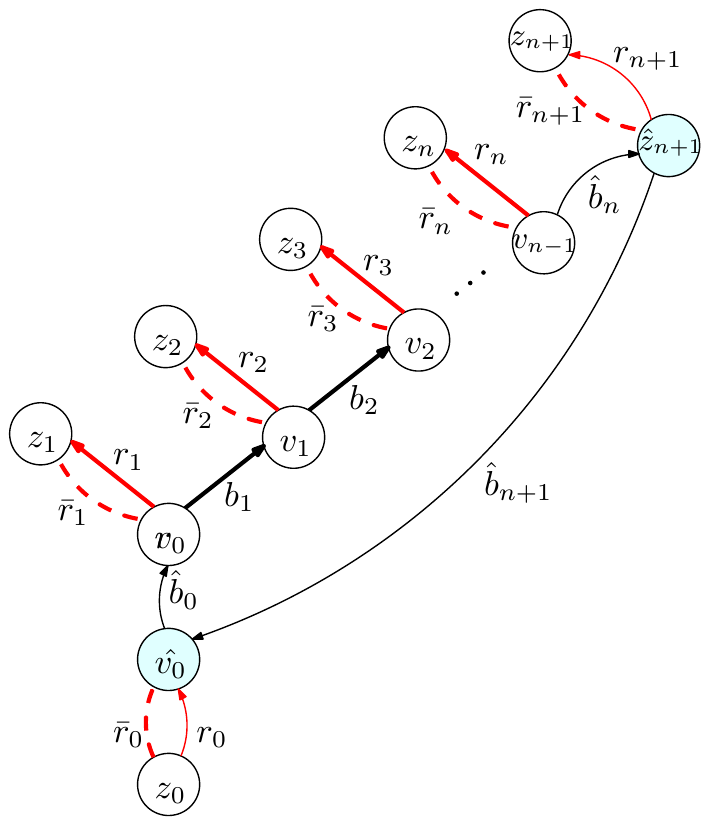}  
 \caption{}
  \label{fig:sub-OrDgraph}
\end{subfigure}
\caption{(a) The decision tree $\cT$ for $G=1$. (b)  The decision graph $\tilde{\cT}$ for $G=1$. 
\\
We note that any internal vertex of the decision tree belongs to a black path. Here, since $G=1$ we have a single black path containing all internal vertices. The subroutine $
\cA_{\BlackPath}$ for the decision tree $\cT$ works as follows: given any vertex $v_i$ outputs $n+1$ as the length of the black path containing $v_i$. Moreover, given any $v_i$ and a number $1\leq k\leq n+1$ outputs the $k$-th vertex of the black path containing $v_i$, which is $v_{k-1}$ if $k\neq n+1$, and $z_{n+1}$ otherwise. 
}
\label{fig:OR}
\end{figure}

Recall that $M$ has a column for every edge, and a column for any leaf and the root in $\cT$ (total of $4n+3$ columns). To obtain a clearer description of $\ker( M)$, which in particular would be useful in the general case in the next subsection, we add new columns to $M$. These columns introduce new vectors in the span program, yet they are carefully chosen in such a way that they change the complexity of the span program only by a constant factor. To keep the graphical picture in mind, we describe these new columns through a modification of the decision tree. In the following we first build a new graph $\tilde \cT$ out of $\cT$ that is no longer a tree, so we call it a decision graph. Next, we describe a matrix $\tilde M$ that is constructed from $\tilde \cT$ and contains $M$ as a submatrix.  

The vertex set of $\tilde{\cT}$ is
\begin{equation*}
\tilde{V}=V\cup \{\hat{v}_0  \cup \hat{z}_{n+1}\},
\end{equation*}
and its edge set is $\tilde{E}  = \tilde{E}^b \cup \tilde{E}^r \cup \tilde{E}^r_{\pseudo}$ with
\begin{align*}
& \tilde{E}^b = \{\hat{b}_0,b_1,\ldots,b_{n-1},\hat{b}_n,\hat{b}_{n+1}\}\\
& \tilde{E}^r = \{r_0,r_1,\ldots,r_{n+1}\}\\
& \tilde{E}^r_{\pseudo} =
\{ \bar{r}_0,\bar{r}_1,\ldots,\bar{r}_{n+1}\}.
\end{align*}
The adjacencies of vertices and edges in $\tilde \cT$ are shown in Figure~\ref{fig:sub-OrDgraph}.
As in $\cT$, any edge of $\tilde \cT$ also takes a color. $\tilde{E}^b$ is the set of black edges, $\tilde{E}^r$ is the set of red edges and $\tilde{E}^r_{\pseudo}$ is the set of \emph{pseudo-edges} that are also in red. As will be clear below, the choice of the term pseudo-edge for the latter edges is due to their associated vectors that are different from those of other edges. Moreover, as vertex $z_0$ is the root of $\cT$ and vertices $\{z_1, \dots, z_{n+1}\}$ are its leaves, for simplicity we call them the root and leaves of $\tilde \cT$ as well. 

We now describe the matrix $\tilde M$. Rows of $\tilde M$ are indexed by vertices of $\tilde \cT$ and its columns are indexed by edges of $\tilde \cT$ and some loops (that are not present in $\tilde M$) associated to the root and leaves. To distinguish rows and columns we denote the index of a column by $\ket{e_\ast}$ if it corresponds to an edge, and by $\ket{\bar e_{\ast}}$ if it corresponds to a pseudo-edge. In particular, we denote the column corresponding to a vertex $z$ (either the root or a leaf) by $\ket{e_z}$. 
Let 
\begin{align}\label{eq:alphabetagamma}
\alpha=\frac{\sqrt{2} \epsilon}{\sqrt{\wsize^+}},\qquad  \beta=\sqrt{W_\black}, \qquad \gamma=\sqrt{W_\red},
\end{align}
where $\epsilon>0$ is a constant, and $\wsize^{+}, W_{\black}, W_{\red}$ are as before. Then columns of $\tilde{M}$ are described in three categories:
\begin{enumerate}
\item $\tilde M$ has a column for the root that contains the vector $\alpha \ket{z_0}$ 
 and a column for each leaf $z$ that contains $-\alpha \ket{z}$. Therefore, $\tilde M\ket{e_{z_0}} = \alpha \ket{z_0}$ and $\tilde M\ket{z_j} = -\alpha \ket{z_j}$ for $1\leq j\leq n+1$. Note that these columns are already present in $M$.
 
\item For any black edge $b=(u,v)$, $\tilde M$ has a column containing the vector $\beta\big(\ket{u}-\ket{v}\big)$,
and for any red edge $r=(u',v')$, it contains the vector $\gamma \big(\ket{u'}-\ket{v'}\big)$. Therefore, $\tilde M \ket{e_b} = \beta\big(\ket{u}-\ket{v}\big)$ and $\tilde M\ket{e_r} = \gamma \big(\ket{u'}-\ket{v'}\big)$.

\item For any pseudo-edge $\bar{r}=(u,v)$, $\tilde M$ contains the vector $\gamma \big(\ket{u}+\ket{v}\big)$, so that $\tilde M \ket{e_{\bar r}} = \gamma \big(\ket{u}+\ket{v}\big)$.  

\end{enumerate} 

Observe that the column associated to a pseudo-edge is a scalar times the sum of its end-points, unlike that of edges that is the difference of its end-points.
The matrix representation of $\tilde M$ is shown in Figure~\ref{fig:Mtilde1}.

\begin{figure}
\centering \includegraphics[scale=1]{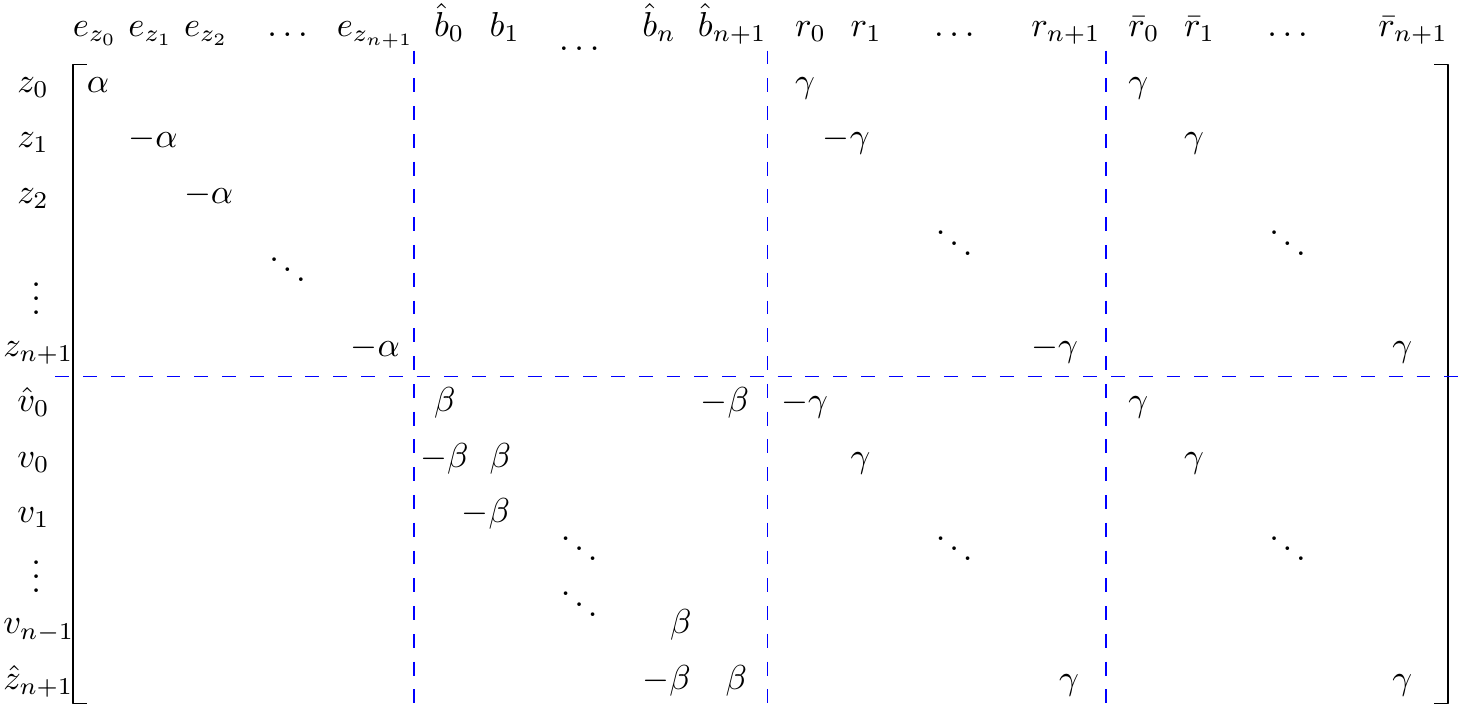}
\figcaption{  Matrix representation of $\tilde{M}$ for $G=1$. Rows of $\tilde{M}$ are indexed by vertices of $\tilde{\cT}$ and its columns are indexed by root, leaves, edges and pseudo-edge of $\tilde{\cT}$. 
 }\label{fig:Mtilde1}
\end{figure}

Before describing the kernel of $\tilde M$ let us briefly explain the span program associated to $\tilde M$. As before, the first $n+2$ columns of $\tilde M$ correspond to target vectors. Indeed, the sum of the first column and the $(j+1)$-th column, $1\leq j\leq n+1$, is the $j$-th target vector. In the decision graph $\tilde \cT$ (as in $\cT$) vertex $v_j$ corresponds to query $x_{j+1}$. Black edges $b_1, \dots, b_{n-1}, \hat b_n$ are labeled by $0$, and red edges $r_1, \dots, r_n$ are labeled by $1$. This means that, e.g., the edge $r_1$ (vector $\tilde M\ket{e_{r_1}} = \gamma\big(\ket{v_0} - \ket{z_1}\big)$) is available if $x_1=1$. We let the edges $r_0, \hat b_0$ and  $r_{n+1}$ be free and \emph{always available}. Furthermore, we assume that the edge $\hat b_{n+1}$ and pseudo-edges $\bar r_0, \dots, \bar r_{n+1}$ are forbidden and \emph{never available}.  Putting these together the span program is fully described. It is not hard to verify that this is a valid span program for the starting function $f$ since as before positive and negative witnesses correspond to paths from the root to the leaves of $\tilde \cT$. The positive witness size of this span program is almost unchanged. However, its negative witness size increases because of the appended never available pseudo-edges, yet the weight of these pseudo-edges are chosen in such a way that the negative complexity increases only by a constant factor. In particular, these weights are equal to the weights of their parallel red edges, so they increase the negative complexity only by a factor of $2$.  Thus, the complexity of the new span program provided by the decision graph $\tilde \cT$ and $\tilde M$ is $O(\sqrt n)$ as before.

We now move to the characterization of $\ker\tilde M$.
For $0\leq i\leq n+1$ let
\begin{align*}
&\ket{{r_i}^-}=\frac{1}{\sqrt 2}\big(\ket{\bar{r}_i}-\ket{r_i}\big)\\
&\ket{{r_i}^+}=\frac{1}{\sqrt 2}\big(\ket{\bar{r}_i}+\ket{r_i}\big).
\end{align*}
Then consider the following sets of vectors:
 \begin{itemize}
 \item[(I)]
 Type I vectors consist of $\frac{-\sqrt 2}{  \alpha}\ket{e_{z_0}}+\frac{1}{\gamma} \ket{{r_0}^+}$ and vectors
 $$\frac{\sqrt 2}{\alpha}\ket{e_{z_i}}+\frac{1}{\gamma} \ket{{r_i}^-}, \qquad 0\leq i\leq n+1. $$
 
\item[(II)] Type II vectors consist of the following vectors:
\begin{align*}
 & \frac{\sqrt 2}{\beta}\ket{\hat{b}_0}-\frac{1}{\gamma}\ket{{r_0}^-}+\frac{1}{\gamma}\ket{{r_{1}}^+}, \\
 &\frac{\sqrt 2}{\beta}\ket{b_i}-\frac{1}{\gamma}\ket{{r_i}^+}+\frac{1}{\gamma}\ket{{r_{i+1}}^+}, \qquad \forall 1\leq i\leq n-1,\\
  &\frac{\sqrt 2}{\beta}\ket{\hat{b}_n}-\frac{1}{\gamma}\ket{{r_n}^+}+\frac{1}{\gamma}\ket{{r_{n+1}}^+}, \\ 
 &\frac{\sqrt 2}{\beta}\ket{\hat{b}_{n+1}}-\frac{1}{\gamma}\ket{{r_{n+1}}^+}+\frac{1}{\gamma}\ket{{r_{0}}^-}.
\end{align*}
\end{itemize}

Straightforward computations show that all these vectors belong to $\ker(\tilde M)$. Moreover, Type I vectors are orthogonal to each other, and to Type II vectors. Moreover, Type II vectors are linearly independent. 
To verify this we may form a matrix $B$ by putting these vectors as its columns, see Figure~\ref{fig:BG1}. Clearly, the second block of $B$ is full-rank, so these vectors are linearly independent.
Thus, Type I and Type II vectors together form $2n+4$ independent vectors in the kernel of $\tilde M$. On the other hand, as is clear from the matrix representation of $\tilde{M}$ in Figure~\ref{fig:Mtilde1}, the last $2(n+2)$ columns of $\tilde M$ are linearly independent. Thus $\tilde M$ is full-rank and its nullity equals $2n+4$. This means that Type I and Type II vectors span $\ker(\tilde M)$. 

To summarize the above findings, the null space of $\tilde M$ consists of two \emph{orthogonal} subspaces spanned by Type I and Type II vectors. Thus the reflection through $\ker(\tilde M)$ can be implemented by composing reflections along each of these subspaces. In the following we show that each of these reflections can be implemented in time $O(\log^2 n)$. 

\begin{figure}
\centering \includegraphics[scale=1]{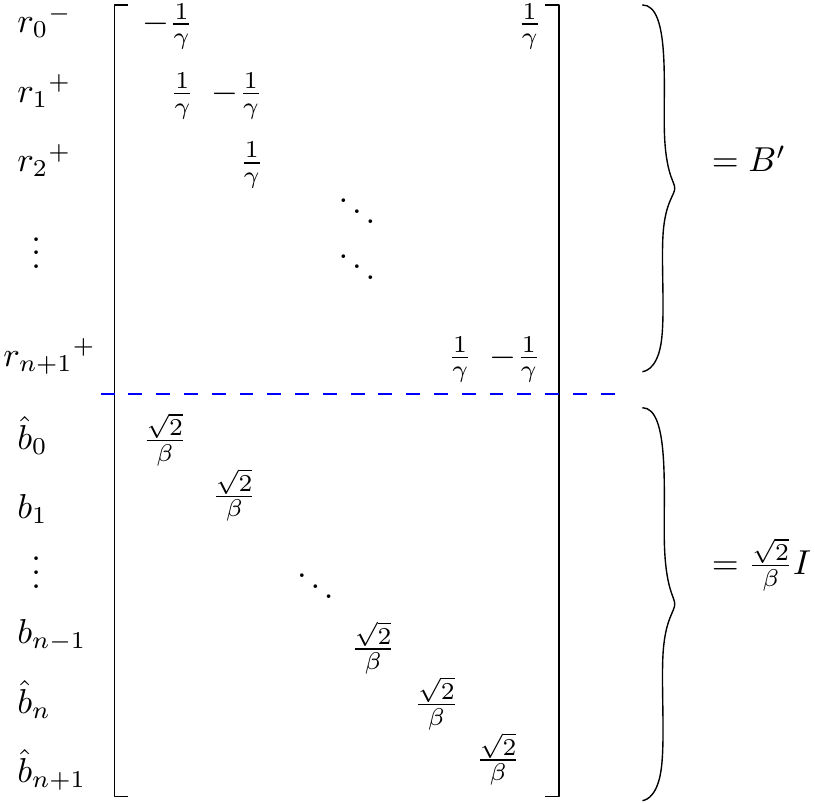}
\figcaption{Matrix representation of the B matrix. 
 }\label{fig:BG1}
\end{figure}

\paragraph{Reflection through Type I vectors:}
Since Type I vectors are orthogonal to each other, reflection through their span is easy.  Indeed, as $\frac{-\sqrt 2}{  \alpha}\ket{e_{z_0}}+\frac{1}{\gamma} \ket{{r_0}^+}$ is orthogonal to other Type I vectors we can implement the reflection through Type I vectors as the composition of the reflection through this vector and reflection through others. The former reflection is easy to implement as it is just a single (sparse) vector. For the later reflection define the $3\times 3$ unitary $K$ by $K\ket 2=\ket 0$ and
\begin{equation*}
 K\ket 0= \sqrt{  \frac{2}{\alpha^2}    +   \frac{1}{\gamma^2}  }\left(\frac{\sqrt 2}{\alpha}\ket 2+\frac{1}{\gamma}\ket{1}\right),\quad   K\ket {1}= \sqrt{  \frac{2}{\alpha^2}    +   \frac{1}{\gamma^2}  }\left(\frac{1}{\gamma}\ket{2}-\frac{\sqrt 2}{\alpha}\ket{1}\right)
\end{equation*}
Also, define $H$ by 
$$H\ket 0=\frac{1}{\sqrt 2} (\ket 0+\ket 1), \qquad H\ket 1=\frac{1}{\sqrt 2}(\ket 0-\ket 1), \qquad H\ket 2=\ket 2.$$
Moreover, consider the relabeling map 
$$P\ket{i}\ket 0 =\ket{\bar r_i}, \qquad P\ket i\ket 1 = \ket{r_i}, \qquad P\ket i\ket 2=\ket{e_{z_i}}.$$ 
Then the desired reflection equals 
\begin{align}\label{eq:PHK-reflection}
P\Big( I_{n+2}\otimes HK\big(2|0\rangle\langle 0|-I\big) K^\dagger H^\dagger\Big) P^\dagger,
\end{align}

and can be implemented in time $O(\log n)$.

\paragraph{Reflection through Type II vectors:} 
This reflection needs more consideration since Type II vectors are not orthogonal to each other. To this end, we construct an orthogonal basis for the span of Type~II vectors. We first extend the matrix $B$ consisting of Type II vectors as in Figure~\ref{fig:BG1}, by adding some column vectors to make an invertible matrix $Q$:
\begin{equation*}
Q=\begin{bmatrix}
B' & \frac{\sqrt 2}{\beta} I \\ 
\frac{\sqrt 2}{\beta} I & -B'^\dagger
\end{bmatrix}.
\end{equation*}
Then we have
\begin{equation*}
Q^\dagger Q=\begin{bmatrix}
B'^\dagger B'+\frac {2}{\beta^2} I & 0 \\ 
0 & B'B'^\dagger+\frac {2}{\beta^2} I
\end{bmatrix} = 
\begin{bmatrix}
L+\frac {2}{\beta^2} I & 0 \\ 
0 & L+\frac {2}{\beta^2} I
\end{bmatrix}, 
\end{equation*}
where
\begin{equation*}
L=B'B'^\dagger=B'^\dagger B'=\frac{1}{\gamma^2}
\begin{bmatrix}
2 & -1 &  &  &  & -1 \\ 
-1 & 2 & -1 &  &  &  \\ 
 & -1 & 2 & \ddots &  &  \\ 
 &  & \ddots & \ddots & -1 &   \\ 
  &   &   & -1 & 2 & -1 \\ 
-1 &   &   &   & -1 & 2
\end{bmatrix}.
\end{equation*}
Since both $B'$ and $L$ are diagonal in Fourier basis we have $F^\dagger B' F=\Lambda_{B'}$ and
\begin{equation*}
\begin{bmatrix}
F^\dagger & 0 \\ 
0 & F^\dagger
\end{bmatrix} 
Q^\dagger Q
\begin{bmatrix}
F & 0 \\ 
0 & F
\end{bmatrix} =
\begin{bmatrix}
\Lambda & 0 \\ 
0 & \Lambda
\end{bmatrix} 
\end{equation*}
where $\Lambda_{B'}$ and $\Lambda$  are diagonal matrices containing eigenvalues of $B'$ and $L+\frac {2}{\beta^2} I$  respectively, and $F$ is the Fourier matrix
$$
F=\frac{1}{\sqrt{n+2}}\sum_{j,k=0}^{n+1}  e^{-2\pi i j k/(n+2)}\ket{j}\bra{k}.
$$
Therefore, the matrix
\begin{equation*}
W:=Q
\begin{bmatrix}
F & 0 \\ 
0 & F
\end{bmatrix} 
\begin{bmatrix}
\Lambda^{-1/2} & 0 \\ 
0 & \Lambda^{-1/2}
\end{bmatrix} 
\end{equation*}
is unitary. On the other hand, by its definition the first $(n+2)$ columns of $W$ are in the span of columns of $B$. This means that these $(n+2)$ columns form an orthonormal basis for the span of columns of $B$. As a result, the reflection through the span of Type~II vectors is equal to
\begin{equation}\label{eq:type2ref}
W
\begin{bmatrix}
I & 0 \\ 
0 & -I
\end{bmatrix} 
W^\dagger.
\end{equation}
Thus, to implement this reflection, it suffices to implement $W$ and $W^\dagger$. 

Using $F^\dagger B' F=\Lambda_{B'}$ we have
\begin{align*}
W= &
\begin{bmatrix}
B' & \frac{\sqrt 2}{\beta} I \\ 
\frac{\sqrt 2}{\beta} I & -B'^\dagger
\end{bmatrix} 
\begin{bmatrix}
F & 0 \\ 
0 & F
\end{bmatrix} 
\begin{bmatrix}
\Lambda^{-1/2} & 0 \\ 
0 & \Lambda^{-1/2}
\end{bmatrix} \nonumber\\
=& 
\begin{bmatrix}
B'F\Lambda^{-1/2} & \frac{\sqrt 2}{\beta} F\Lambda^{-1/2} \\ 
\frac{\sqrt 2}{\beta} F\Lambda^{-1/2} & -B'^{\dagger}F\Lambda^{-1/2}
\end{bmatrix}  
\nonumber
\\ = & 
\begin{bmatrix}
F & 0 \\ 
0 & F
\end{bmatrix} 
\begin{bmatrix}
\Lambda_{B'}\Lambda^{-1/2} & \frac{\sqrt 2}{\beta}\Lambda^{-1/2} \\ 
\frac{\sqrt 2}{\beta}\Lambda^{-1/2} & -\Lambda^\dagger_{B'}\Lambda^{-1/2}
\end{bmatrix}.
\end{align*}
The first matrix in this decomposition is the tensor product of the $2\times 2$ identity matrix with Fourier transform that is implementable in time $O(\log^2 n)$~\cite{NC}.\footnote{We note that the circuit for the Fourier transform in~\cite{NC} works only for powers of two. To circumvent this, we may insert some dummy black edges (with weight $W_{\black}$) in the graph to make the length of all black paths be powers of two. This would increase $T$ by at most a factor of $2$, and does not change the order of the complexity of the algorithm. We note that, this change would enforce modifications in subroutines $\cA_{\Local}$ and $\cA_{\BlackPath}$, that can easily be taken care of.}
Since $\Lambda$ and $\Lambda_{B'}$ are diagonal, rearranging rows and columns of the second matrix, we find that it is a block-diagonal matrix with blocks of size $2$. Then, as entries in each block can be explicitly computed, this matrix can also be implemented in time $O(\log n)$ with $O(1)$ queries given access to QRAM.
More precisely, given an index of a $2\times 2$ block of this matrix, we can compute its entries in time $O(\log n) $ and implement it. Hence, if the index of the block is given in superposition via QRAM, we can implement the whole unitary in time $O(\log n)$.
We conclude that $W$ and similarly $W^\dagger$ can be implement in time $O(\log^2 n)$.

\medskip
Putting all these together we find that the reflection through $\ker(\tilde M)$ can be implemented in time $O(\log^2 n)$. On the other hand, it is not hard to verify that $R_{\Pi}$ can also be implemented in time $O(\log n)$. Then by Theorem~\ref{thm:QQAlg}, $f$ can be computed in time $O(\log^2 n\sqrt n)$.

\subsection{$G> 1$}\label{ssec:G>1}

We can think of a decision tree with $G> 1$ as a union of decision trees with $G=1$. The point is that to construct such a decision tree we can start with a single red edge, and recursively proceed as follows: whenever we see a red edge attach a decision tree with $G=1$ to it as a child subtree. This way of thinking of a decision tree for arbitrary $G$ enables us to use ideas from the case of $G=1$ to prove Theorem~\ref{thm:binaryClassical2quantumtime}. To this end, as in the case of $G=1$, we need to add some \emph{returning} black edges to the decision tree to turn black paths to black cycles. This changes the decision tree $\cT$ to a \emph{decision graph} $\tilde \cT$.
Then,  $\ker(\tilde M)$ would be a direct sum of orthogonal subspaces each of which corresponds to a black cycle, and have the same structure as  the kernel in the case of $G=1$. Then reflection through $\ker(\tilde M)$ can be implemented using similar ideas (via Fourier transform) as before. As the details are explained shortly, this gives the proof of Theorem~\ref{thm:binaryClassical2quantumtime} in the case of $\ell=2$.

The case of $\ell>2$ is slightly different. The point is that, as mentioned before, the NBSP in the case of $\ell>2$ is  a bit different from that of the binary case (see Appendix~\ref{app:nonbinary}), and we do not know if such a span program can directly be compiled without that much overhead in the time complexity. Thus, our strategy to prove the non-binary input alphabet case, is to transform a non-binary decision tree to a binary one. This transformation is in such a way that it changes depth of the tree from $T$ to $O\big(T+\log(\ell)G\big)$, and changes $G$ to $\log(\ell)G$. 
Then using the theorem in the case of $\ell=2$, the desired result follows. 
See Appendix~\ref{app:Non-binary_Implement} for more details on the case of $\ell>2$.


\medskip
Now we focus on the case $\ell=2$ and $G>1$.
Our first step is to construct the decision graph $\tilde \cT$ from $\cT$. We first fix some notation. Let $\cT=(V, E)$ where $V$ is the vertex set and $E$ is the edge set of $\cT$. Then let

\begin{itemize}
\item $V_{\leaf}$ be the set of leaves of $\cT$.
\item $V_\internal$ be the set of internal vertices of $\cT$.
\item $V_{\leaf}^b$ be the set of leaves $z$ whose parent edge is black.
\item $V^r_{\leaf}$ be the set of leaves $z$ whose parent edge is red.
\item $ V^r_\internal$ be the set of internal vertices $v$ whose parent edge is red.
\item $E^r$ be the set of red edges. 
\item $E^b=E\setminus E^r$ be the set of black edges. 
\item $v_p$ for a vertex $v$ be the parent vertex of $v$ in $\cT$.
\item $z_{v}\in V_{\leaf}^b$ be the black leaf in $\cT$ that is connected to $v$ via a path of black edges.
\end{itemize}

As in the case of $G=1$ it is convenient to assume that the root of $\cT$ is a vertex $z_0$ that is connected to an internal vertex $v_0$ (the previous root) via a \emph{red} edge (with weight $\sqrt{W_{\red}}$) that is always available. As in the case of $G=1$ this modification does not significantly change the complexity of the associated span program.

\begin{figure}
\centering \includegraphics[scale=.9]{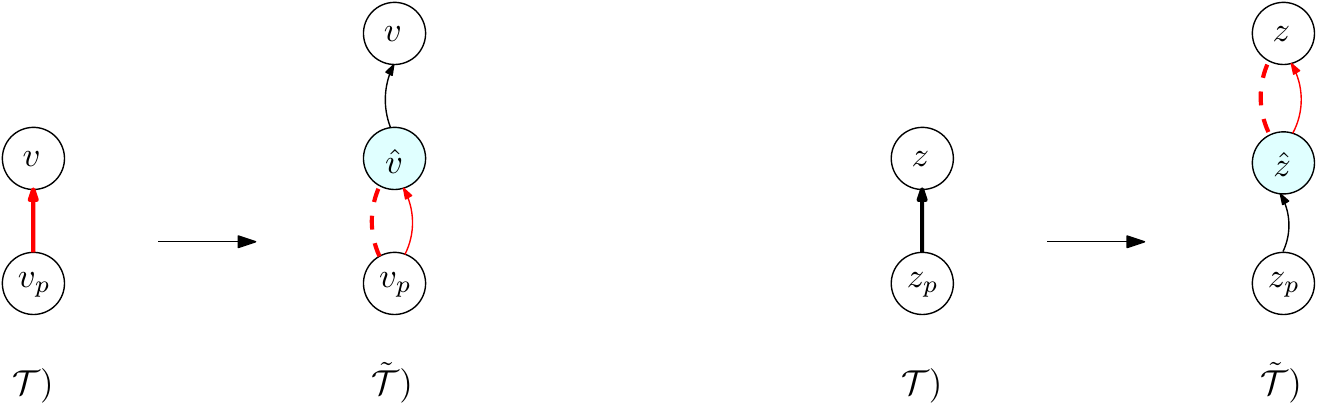}\hspace{1.in}
\figcaption{ Any red edge in $\cT$ introduces two new edges (included in $\tilde{E}_2$ and $\tilde{E}_3$) and a pseudo-edge in $\Tilde{\cT}$. Any leaf in $\cT$  which is at the end of a black path,  introduces two new edges (included in $\tilde{E}_4$ and $\tilde{E}_5$) and a pseudo-edge in $\Tilde{\cT}$    
 }\label{fig:local-modifications}
\end{figure}

We now describe the decision graph of $\tilde \cT$. This graph is obtained by applying three modifications on $\cT$. Two of these modifications are explained in Figure~\ref{fig:local-modifications}. The third one is adding a black edge from the end-point of a black path to its starting vertex. An example of the construction of $\tilde \cT$ from $\cT$ is given in Figure~\ref{fig:TTprime}. More formally,
the vertex set of $\tilde{\cT}$ is 
$$\tilde{V}=V\cup \big\{\hat{z}: z\in V^b_{\rm leaf}\big\} \cup\big\{\hat{v}:v\in V^r_\internal \big\} .$$
Its edge set $\tilde{E}$ consists of the following sets of edges and pseudo-edges:
\begin{enumerate}
\item $\tilde E_1=\big\{(w,v)\in E:\,  v\notin V^r_\internal \cup V^b_{\leaf}\big\}$. The color of these edges is the same as their color in $\cT$.
\item $\tilde E_2=\{(\hat{v},v): v\in V^r_\internal\} $. All these edges take black color.
\item $\tilde E_3=\{(v_p,\hat{v}): v\in V^r_\internal\}$. All these edges take red color.
\item $\tilde E_4=\{(\hat{z} , z): z\in V^b_{\leaf}\}$. All these edges take red color.
\item $\tilde E_5=\{(z_p,\hat{z}): z\in V^b_{\leaf}\}$.  All these edges take black color.
\item $\tilde E_6=\{(\hat{z_v},\hat{v}): v\in V^r_\internal\}$. All these edges take black color.
\item  $\tilde{E}_\pseudo=\big\{\overline{(v,w)}: (v,w)\in \tilde{E}^r\big\}$. These are pseudo-edges with red color.
\end{enumerate}
As before we decompose the edge set of $\tilde \cT$ as $\tilde E^r\cup \tilde E^b\cup \tilde E_{\pseudo}$ where $\tilde E^r$ is the set of red edges, $\tilde E^b$ is the set of black edges and $\tilde E_{\pseudo}$ is the set of red pseudo-edges.

These modifications on $\cT$, resulting in $\tilde \cT$, simplify the structure of the null-space of $M$. In particular, insertion of the pseudo-edges  enables us to think of the graph $\tilde \cT$ as a disjoint union of subgraphs (cycles) consisting of black edges. This becomes clear below.

We now define the matrix $\tilde M$ similar to the case of $G=1$. As before, rows of $\tilde M$ are indexed by vertices of $\tilde \cT$. Moreover, columns of $\tilde M$ are indexed by edges and pseudo-edges of $\tilde \cT$ plus its root and leaves. $\tilde M$ can explicitly be written as 

\begin{figure}
\centering \includegraphics[scale=.75]{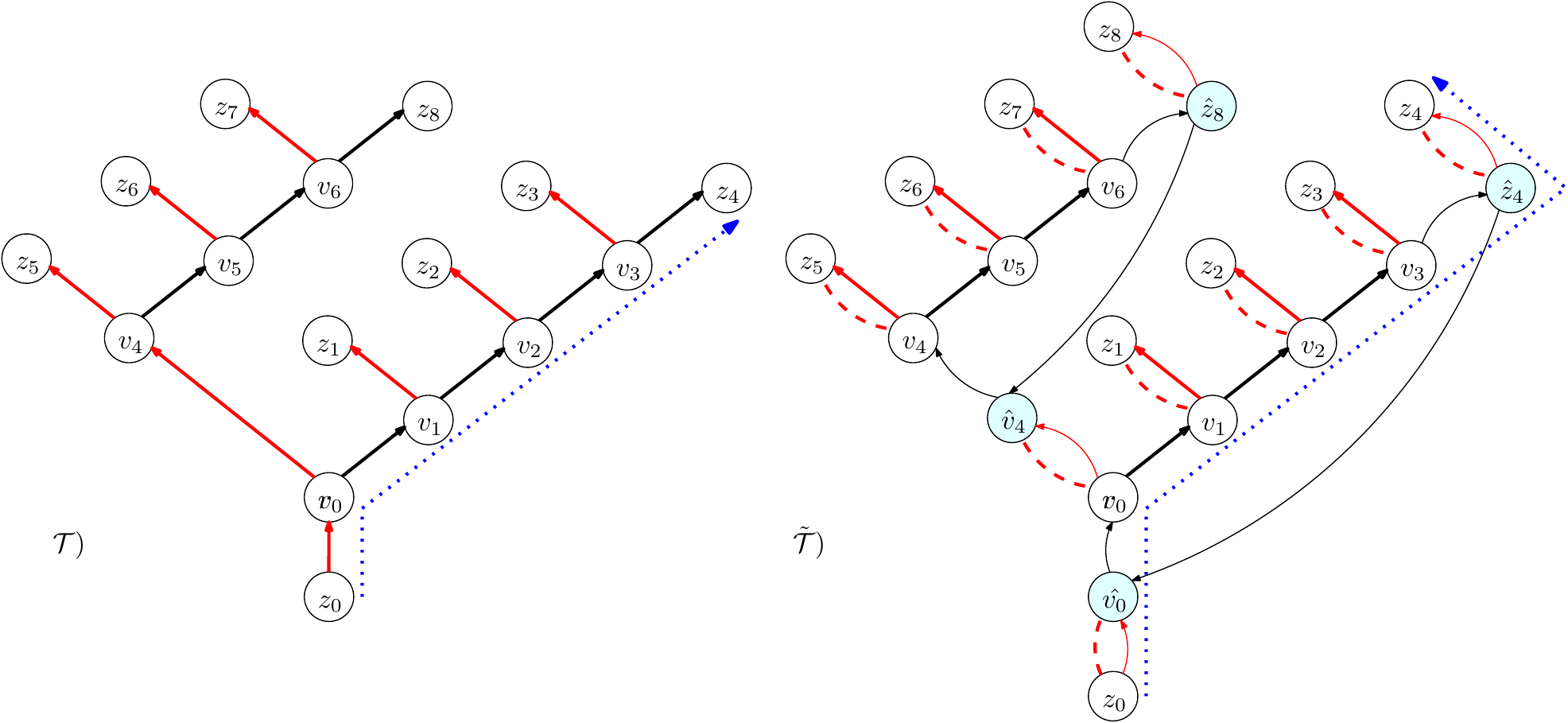}
\figcaption{An example that shows how we build $\tilde{\cT}$ out of $\cT$.  To build $\tilde{\cT}$, we start with $\cT$ and  change the graph using the structure we depicted in Figure~\ref{fig:local-modifications}. We also add a pseudo-edge to any red edge. Pseudo-edges are shown using dashed lines. For any path from the root to a leaf in $\cT$, there exists a path from the root to the same leaf in $\tilde{\cT}$. The dotted line represents one of these paths.
 }
\label{fig:TTprime}
\end{figure}

\begin{align*}
\tilde{M}=&\alpha \ket{z_0}\bra{e_{z_0}}-\alpha\sum_{z\in  V_{\leaf}}\ket{z}\bra{e_z}
+\beta\sum_{(v,w)\in \tilde{E}^b}(\ket{v}-\ket{w})\bra{e_{vw}}\\
&+\gamma\sum_{(v,w)\in \tilde{E}^r} \big(\ket{v}-\ket{w}\big)\bra{e_{vw}}
 +\gamma \sum_{\overline{(v, w)}\in \tilde E_{\pseudo}} \big(\ket{v}+\ket{w} \big)\bra{\bar{e}_{vw}}, 
\end{align*}
where $\alpha, \beta, \gamma$ are defined in~\eqref{eq:alphabetagamma}.

Next we describe the associated span program based on $\tilde M$. As before, target vectors correspond to the sums of the first column of $\tilde M$ associated to its root, and the columns associated to the leaves, i.e., $\ket{z_0} - \ket{z}$ for any leaf $z$. The remaining columns of $\tilde M$ (associated to edges and pseudo edges of $\tilde \cT$) determine the available and unavailable input vectors as follows:

\begin{itemize}
\item $(w,v)\in \tilde E_1$ is available whenever the associated edge in $\cT$ is available. 
\item $(\hat{v},v)\in \tilde E_2$ is available whenever $(v_p, v)$ is available in $\cT$.
\item $(v_p,\hat{v})\in \tilde E_3$ is available whenever $(v_p, v)$ is available in $\cT$.
\item $(\hat{z} , z)\in \tilde E_4$ is available whenever $(z_p, z)$ is available in $\cT$.
\item $(z_p,\hat{z})\in \tilde E_5$ is available whenever $(z_p, z)$ is available in $\cT$.

\item Edges in $\tilde E_6$ and  $\tilde{E}_\pseudo$ are always \emph{unavailable} and forbidden. 

\end{itemize}

We now have a complete description of the span program. It is not hard to verify that this span program computes the same function as the span program associated to $\cT$. 
We then examine the complexity of this span program. 
For the  positive complexity, notice that for any path from the root to a leaf $z$ in $\cT$ of available edges there is a path form the root to the same leaf in $\tilde \cT$ with some newly added edges that are all available. These new edges increase the positive complexity by at most $\frac{1}{W_\black}(G+1)=G(G+1)\leq GT$. Moreover, the negative complexity is increased in $\tilde \cT$ by a term smaller than $W_\red (T+1)=\frac{T+1}{T}\leq 2$. Thus the total complexity of the span program remains the same up to a constant factor.

\paragraph{Kernel of $\tilde{M}$:}
We can now use the structure of the new span program based on $\tilde \cT$ to present a concise description of $\ker(\tilde M)$. As in the case of $G=1$, a basis for $\ker\tilde M$ consists of two types of vectors. Type I vectors correspond to the root and leaves of $\tilde \cT$. Type II vectors come from cycles in $\tilde \cT$ of black edges; any black cycle gives us a collection of Type II vectors in $\ker(\tilde M)$. 

\begin{lemma}\label{lem:kernel} 
For any red edge $(v, w)$ in $\tilde \cT$ let $$\ket{{e_{uw}}^-}=\frac{1}{\sqrt 2}\big(\ket{\bar{e}_{uw}}-\ket{e_{uw}}\big), \qquad \ket{{e_{uw}}^+}=\frac{1}{\sqrt 2}\big(\ket{\bar{e}_{uw}}+\ket{e_{uw}}\big).$$ 
Then the following sets form a basis for the kernel of $\tilde{M}$:
\begin{itemize}
    \item Type I vectors are $-\frac{\sqrt 2}{\alpha}\ket{e_{z_0}} + \frac{1}{\gamma}\ket{e_{z_0\hat v_0}^+}$ and 
$$\frac{\sqrt{2}}{\alpha}\ket{e_z}+\frac{1}{\gamma }\ket{{e_{z_pz}}^-}, \qquad\forall z\in \tilde{V}_{\leaf}.$$

\item Type II vectors consist of the following collection of vectors for any black cycle in $\tilde \cT$. Let   
$(\hat{v}_0, v_0, v_1, \dots, v_k, v_{k+1}=\hat z)$ be a black cycle. Let $(v_i, v'_i)$, for $0\leq i\leq k$ be the red edge connected to $v_i$. Also let $(v', \hat v_0)$ and $(\hat z, z)=(v_{k+1}, v_{k+1}')$ be the red edges connected to $\hat v_0$ and $\hat z=v_{k+1}$ respectively. Then the following vectors belong to $\ker(\tilde M)$:

\begin{align*}
&\frac{\sqrt 2}{\beta}\ket{e_{\hat{v}_0v_0}}     -\frac{1}{\gamma}\ket{{e_{v'\hat{v}_0}}^-}+   \frac{1}{\gamma}\ket{{e_{v_0v'_0}}^+}\\ 
&\frac{\sqrt 2}{\beta}\ket{e_{v_{i-1}v_{i}}}    -\frac{1}{\gamma}\ket{{e_{v_{i-1}v'_{i-1}}}^+}+   \frac{1}{\gamma}\ket{{e_{v_iv'_i}}^+}, \qquad 1\leq i\leq k+1 \\ 
&
 \frac{\sqrt 2}{\beta}\ket{e_{\hat{z}\hat{v}_0}} -\frac{1}{\gamma}\ket{{e_{\hat{z}z}}^+}+\frac{1}{\gamma}\ket{{e_{v'\hat{v}_0}}^-} .
\end{align*}

\end{itemize}

\end{lemma}

\begin{proof}
It is easy to see that all of these vectors belong to $\ker(\tilde M)$. Moreover, they are independent. Indeed, Type I vectors are orthogonal to each other and to Type II vectors. Also, Type II vectors associated to different black cycles are orthogonal. Then it suffices to verify that Type II vectors associated to a cycle are independent. Next, as in the case of $G=1$, considering columns of $\tilde M$ associated to red edges and red pseudo-edges, we find that $\tilde M$ is onto. Then by the rank-nullity theorem, the nullity of $\tilde M$ equals the number of columns of $\tilde M$ other than those associated with red edges and pseudo edges. These columns come from either the root and leaves, or black edges. Now observe that Type I vectors come from the root and leaves, and Type II vectors come from black edges. We conclude that Type I and Type II vectors span $\ker(\tilde M)$.
\end{proof}

As mentioned above all Type~I vectors are orthogonal to each other and to Type~II vectors. Also, Type~II vectors associated to different black cycles are orthogonal to each other. Nevertheless, as in the case of $G=1$, Type~II vectors for a single black cycle are not orthogonal to each other. Then to implement the reflection through $\ker\tilde M$ we need to construct an orthogonal basis for the subspace spanned by such vectors. This can be done by Fourier transform similarly to the case of $G=1$. 

\subsection{Proof of Theorem~\ref{thm:binaryClassical2quantumtime}}

We can now put all the ingredients together to finish the proof of Theorem~\ref{thm:binaryClassical2quantumtime}. As mentioned before the state space of our quantum algorithm is the space of columns of $\tilde M$. We use an encoding of the columns in such a way that implementation of subroutines $\cA_{\Local}$ and $\cA_{\BlackPath}$ is not costly, and use the same encoding to implement the quantum algorithm. Here, we should note that 
the subroutine $\cA_\Local$ in the statement of the theorem gives information about $\cT$, not $\tilde \cT$ which determines the state space of the algorithm.
Nevertheless, as $\tilde \cT$ is obtained from $\cT$ via simple moves, $\cA_{\Local}$ and $\cA_{\BlackPath}$ can be used to implement $\cA_{\Local}$ on $\tilde \cT$. The same holds for $\cA_{\BlackPath}$ on $\tilde \cT$ which gives information about black cycles as apposed to black paths. 

We can now describe implementations of each of the required reflections for the proof of Theorem~\ref{thm:binaryClassical2quantumtime}.  

\paragraph{Reflection through Type~I vectors:} 
For any $z\in V_\leaf$ we need to apply the reflection through $\frac{\sqrt{\wsize^+}}{\epsilon}\ket{e_z}+\frac{1}{\sqrt{W_\red} }\ket{{e_{z_pz}}^-}$. To this end, we use the subroutine $\cA_\Local$ to decide whether we are currently at a leaf of $\tilde{T}$ or not. If yes, we apply the required reflection similarly to~\eqref{eq:PHK-reflection} in the case of $G=1$. Reflection through $\frac{-\sqrt{\wsize^+}}{\epsilon}\ket{z_0}+\ket{{e_{z_0\hat{v}_0}}^+}$, where $z_0$ is the root, is implemented in the same way. Thus reflection through Type~I vectors can be implemented in time $O(C+\log n)$, where
$C=C_\Local+C_\BlackPath$.

\paragraph{Reflection through Type~II vectors:} 
Recall that Type~II vectors consist of a collection of vectors for any black cycle, and these collections are orthogonal for different cycles. Thus we can implement reflection through Type~II vectors by 
applying a \emph{controlled-reflection} along each of these collections. To this end, using the subroutine $\cA_{\BlackPath}$ we first compute the length of the cycle containing the current vertex. Then using $\cA_{\BlackPath}$ we compute the index of that vertex in the cycle. Next, using the same ideas as in 
Subsection~\ref{ssec:G=1}, in particular equation~\eqref{eq:type2ref}, we apply reflection through the collection of vectors in $\ker(\tilde M)$ associated to that cycle. Finally, we once again use $\cA_{\BlackPath}$ to map indices of the vertices in that cycle to their associated states. 
As the Fourier transform of order $t$ can be implemented in time $O(\log^2 t)$, we conclude that the refection through Type~II vectors can be implemented in time $O(C+\log^2 n)$.

\paragraph{Implementing $R_\Pi$:} To implement
$R_\Pi$ we  use $\cA_{\Local}$ to determine whether the current edge is always available or not. We also use $\cA_{\Local}$ to determine the query index associated to the current edge. Using a quantum query we check whether the current edge matches the answer to that query. If that edge is never available, or the answer to that query does not match the edge, we apply a $(-1)$ phase. Then we uncompute everything including the query answer. This gives an implementation of $R_\Pi$ using two quantum queries in time $O(C+\log n)$.

\begin{remark}
 It has been discussed in~\cite{BT20} that for some functions, we may have more speed-up if we assign different weights to different edges of the decision tree. In the proof of Theorem~\ref{thm:binaryClassical2quantumtime}, we assumed that all black edges have a common weight of $W_\black$. It is easy to see that this proof is also valid if we assign a single weight to all black edges in a black path.   
\end{remark}

\section{Proofs of Propositions~\ref{pro:BFS-based}-\ref{pro:bipartite} }\label{sec:applications-bfs}
Before stating the proofs, we show that the BFS algorithm,  a subroutine used, e.g., in all classical algorithms for Proposition~\ref{pro:BFS-based}, has quantum complexity $O(n^{3/2}\log^2 n)$ in the adjacency matrix model.

Algorithm~\ref{alg:BFS} describes the BFS algorithm.
Line~\ref{algline:BFS-addE_S} of the Algorithm~\ref{alg:BFS} determines the G-coloring of the decision tree. Any red edge of the decision tree is associated to an edge $(u,v)$ in the input graph, in which $L(v)=0$. 
Any black edge of the decision tree is associated to absence of an edge in the input graph or $L(v)=1$.

The state of the BFS algorithm consists of four lists $V$, $L$, $Q$, $E_S$. We need to encode these lists in such a way that implementation of the Local subroutine $\cA_{\Local}$ and the BlackPath subroutine $\cA_{\BlackPath}$ are not costly.
Thus we encode the state of this algorithm in
the following way:
\begin{itemize}
\item[$V$:] Using a list of $n$ elements and a pointer to the current position of $V$ that we are processing in the classical algorithm (the value of $i$ in the loop in Algorithm~\ref{alg:BFS}).
\item[$L$:] Using a list of $n$ numbers. All elements of this list are set to 0 in the beginning of the algorithm. During the algorithm whenever we remove the $i$-th element from $V$, we should set $L[i]=1$ and keep the list $V$ unchanged.
\item[$Q$:] Using a list of $n$ elements and two pointers $q_f$ and $q_l$ to its first and last elements.
\item[$E_S$:] Using a list  of $n$ elements and store a pointer to the last element that we added to this list. This pointer helps us adding new elements in cost $O(\log n)$.
\end{itemize}
 
Note that as long as we are going through black edges in the decision graph for BFS algorithm, the lists $L$, $Q$ and $E_S$ are unchanged;  the only changes we observe in the state of the BFS algorithm is in the pointers to the first element of the queue $Q$ and the pointer to the current element of $V$. Similarly for implementing other parts of subroutines $\cA_\Local$ and $\cA_{\BlackPath}$, it is easy to see that we only need to update the values of some of the pointers, or read the value of a list that its associated pointer refers to. Thus both of these subroutines are implementable in time $O(\log n)$.

In particular we determine that $\cA_\Local$ and $\cA_{\BlackPath}$ subroutines can be computed in time $O(\log n)$. For $\cA_\Local$ we compute the desired outputs in the following way:

\begin{itemize}
    \item The type of the given vertex is found by investigating the pointers, e.g., if all pointers are pointing to the first position of their associated lists then we are in the root.
    \item We can obtain the state of the parent and children of $v$ in logarithmic time by appropriate changes of the pointers and lists. 
    For determining the color of neighboring edges, note that a black edge corresponds to absence of an edge in the graph, for which we only update the pointer to the list $V$. A red edge on the other hand means that a vertex is added to the BFS tree, for which we update $Q$ and $N$ appropriately, and also update the pointer to $V$.
    \item The pointers to $Q$ and $V$ determine the index we query in any vertex. For determining the edge corresponding to a given query answer, if the answer is yes, then the edge color is red, otherwise it is black.
\end{itemize}
And for $\cA_{\BlackPath}$:
\begin{itemize}
    \item The \emph{length} of the black path containing the given vertex is $n-1$.
    \item Given $k$, the $k$-th vertex of the black path containing the given vertex can be computed by increasing the value of the pointer to $v$ by $k$.  
\end{itemize}

The last step is to show that $\cA_\PostProcess$ does not increase the cost of the algorithm from $O(n^{3/2}\log^2 n)$. The edge set of the BFS tree is stored in $E_S$, so we can output the result in time proportional to the size of this list, i.e., $O(n\log n)$, which is negligible compared to $O(n^{3/2}\log^2 n)$.

\begin{algorithm}[ht]
\caption{ BFS$(\cG)$: breadth first search algorithm on graph $\cG$}
\label{alg:BFS}
\begin{algorithmic}[1]
\State Let $V$ be the list of vertices of $\cG$ and $L$ be an array of size $n$, that determines the list of processed vertices and $Q$ be a first in first out queue. \label{algline-queue}
\State $L\leftarrow$ all $0$ array, $Q=\emptyset$, $E_\cS=\emptyset$ \Comment $E_\cS$ stores the edge set of the BFS tree.
		\State add $V(1)$ to $Q$
		\State $L(1)\leftarrow 1$
		\While{$Q \neq \emptyset$}
					\State $u\leftarrow {\rm dequeue}(Q)$
					\For{$i=1$ to $n$}
						\State Query $\big(u,V(i)\big)$
						\If{$\big(u,V(i)\big)\in E(\cG)$ and $L(i)=0$}\label{algline:BFS-addE_S}
							\State add $\big(u,V(i)\big)$ to $E_\cS$
							\State add $V(i)$ to $Q$  \label{algline:BFS-add2Q}
							\State $L(i)\leftarrow 1$		
						\EndIf
					\EndFor
		\EndWhile
\State \Return the BFS forest $\cS=\big(V,E_\cS\big)$
\end{algorithmic}
\end{algorithm}

\begin{proof}[Proof of Proposition~\ref{pro:BFS-based}]
To prove this proposition using  Theorem~\ref{thm:binaryClassical2quantumtime} it is enough to show that $\cA_\Local$ and $\cA_\BlackPath$ associated to the classical algorithms that have been used in the proof of Proposition 9 of~\cite{BT20} have time complexity $O(\log n)$. All of these algorithms are started by running the BFS algorithm.

In all classical algorithms for parts (i) through (vi), after the BFS algorithm according to the description of their algorithms, we search for a special edge in the input graph. This makes minor changes in the decision tree of the BFS algorithm without increasing the costs of $\cA_\Local$ and $\cA_\BlackPath$. We elaborate these minor changes for the problem of detecting bipartiteness. To decide whether $\cG$ is not bipartite we need to find an edge between two vertices whose heights in the BFS tree have the same parities; this reveals an odd-length-cycle. To this end, we need another length-$n$ list $H$, besides $V, L, Q$ and $E_S$ to save the parity of the heights of vertices in the BFS tree. During the BFS algorithm as we add vertices to the BFS tree we also set their heights from the height of their parents in the tree. After completing the BFS tree, we search for an edge between two vertices whose heights have the same parties by scanning the list $H$. These modifications does not increase the costs of the subroutines. This is the case for other parts of the proposition as well (see~\cite{BT20} for more details on these algorithms).

We also need to show that $A_\PostProcess$ costs less than $O(n^{3/2}\log^2 n)$, which is easy to confirm in all of the above cases.
Thus, using Theorem~\ref{thm:binaryClassical2quantumtime} the resulting time complexity would be $O(n^{3/2}\log^2 n)$ .
\end{proof}

\begin{proof}[Proof of Proposition~\ref{pro:DFS-based}]
Proof of this proposition is very similar to that of Proposition~\ref{pro:BFS-based}. The main difference is that the classical algorithm for these problems are based on the DFS algorithm.  The DFS algorithm is similar to the BFS algorithm, the only difference is that instead of the queue $Q$ which is a first-in first-out list, DFS algorithm uses a stack which is a last-in first-out list. Thus its time-efficient implementation is possible using the same ideas that we used for the BFS algorithm.
\end{proof}

\begin{proof}[Proof of Proposition~\ref{pro:BFS-list-based}]
Proof of this proposition is again similar to that of Proposition~\ref{pro:BFS-based}. The only difference is that using the same arguments we show that the time complexity of the BFS algorithm is $O(n\sqrt{m+n}\log^{5/2} n)$ in the adjacency list model.
\end{proof}

\begin{algorithm}[ht]
\caption{ Maximal Matching on a graph $\cG$}
\label{alg:maximal}
\begin{algorithmic}[1]
\State Let $V=\{v_1,\ldots,v_n\}$ be the list of vertices of $\cG$ and $M$ be an array of size $n$, that determines the list of vertices that have been added to the matching.
\State $M\leftarrow$ all $0$ array, $E_\cM=\emptyset$ \Comment $E_\cM$ stores the edge set of the maximal matching.
    \For{$i=1$ to $n$}
        \For{$j=1$ to $n$}
            \State Query $(v_i,v_j)$
            \If{$(v_i,v_j)\in E(\cG)$ and $M(i)=0$ and $M(j)=0$} \label{algline:maximal-if}
                \State add $(v_i,v_j)$ to $E_{\cM}$
                \State $M(i) \leftarrow 1$
                \State $M(j) \leftarrow 1$
            \EndIf
        \EndFor
    \EndFor    
\State \Return the maximal matching $E_\cM$
\end{algorithmic}
\end{algorithm}

\begin{proof}[Proof of Proposition~\ref{pro:bipartite}]
To prove this proposition using Theorem~\ref{thm:binaryClassical2quantumtime}, we use the classical algorithm of Hopcroft and Karp~\cite{HK73}. As of the BFS algorithm, we need to increase the height of the decision tree in order to be able to implement the traverse of any black path in the decision tree. This increases the query complexity of this algorithm compared to those of~\cite{BT20,LL16}. The rest of the proof is similar to Proposition~\ref{pro:BFS-based}.
\end{proof}

\begin{proof}[Proof of proposition~\ref{pro:maximal}]
In the matrix model, we use Algorithm~\ref{alg:maximal} as a classical algorithm for this problem. For the coloring of the associated decision tree, edges of the decision tree that satisfy the condition in Line~\ref{algline:maximal-if} of the algorithm take the red color, all other edges are black. For this decision tree $T=n^2$ and $G\leq \frac n 2$. Also it is easy to verify that $\cA_\Local$ and $\cA_\BlackPath$ have logarithmic costs and $\cA_\PostProcess$ costs less than $O(n^{3/2}\log^2 n)$. Thus, the time complexity of this algorithm is $O(n^{3/2}\log^2 n)$. The proof for the adjacency list model is similar.  
\end{proof}

\bibliography{references}
\bibliographystyle{alpha}


\appendix


\section{Proof of Theorem \ref{thm:QQAlg}}
The following auxiliary lemmas are used to prove Theorem~\ref{thm:QQAlg}.\label{app:QQAlg}

\begin{lemma}[Effective spectral Gap Lemma~\cite{LMRSS11}]\label{lem:eff-esp-gap}
Let $U=(2\Pi-I)(2\Lambda-I)$, where $\Lambda$ and $\Pi$ are two projections.  and let $P_\Theta$ be the orthogonal projection onto span$\{\ket{u}:U\ket{u}=e^{i\theta}\ket{u},|\theta|\leq \Theta\}$. Then, if $\Lambda\ket{u}=0$, we have $||P_\Theta \Pi\ket{u}||\leq \frac{\Theta}{2}||\ket{u}||.$
\end{lemma}

\begin{lemma}[Phase detection~\cite{kit95}]\label{thm:phase-det}
For any unitary $U\in \mathcal{L}(\mathcal{H})$ and $\Theta,\delta>0$, there exists $b=O(\log{1/\Theta}\log{1/\delta})$ and a quantum circuit $R(U)$ on $\mathcal{H}\otimes(C^2)^{\otimes b}$ that is constructed uniformly in $\Theta , \delta$ independent of $U$. $R(U)$ makes at most $O(\frac{\log1/\delta}{\Theta})$ controlled calls to $U$ and $U^{-1}$ and for any eigenstate $(\ket{\beta},e^{i\theta})$ of $U$ with $\theta\in(\-\pi,\pi]$,
\begin{itemize}
\item if $\theta=0$, then $R(U)\ket{\beta}\ket{0^b}=\ket{\beta}\ket{0^b}$ and
\item if $|\theta|>\Theta$, then $R(U)\ket{\beta}\ket{0^b}=-\ket{\beta}\left(\ket{0^b}+\ket{\delta_\beta}\right)$, for some $\ket{\delta_\beta}$ having $\|\ket{\delta_\beta}\|<\delta$. Therefore for every $\ket{\nu}\in \mathcal{H}$ that is orthogonal to all eigenvectors of $U$ having eigenvalue $e^{i\theta}$ for $\theta\leq \Theta$, we have $\|(R(U)+I)\ket{\nu}\ket{0^b}\|<\delta.$
\end{itemize} 
\end{lemma}

Now we present the proof of Theorem~\ref{thm:QQAlg}. Let $U=(2\Pi_x-I)(2\Lambda-I)$ and  
$\ket{\phi_{x\pm}}=\frac{1}{\sqrt 2}\big(\ket{0} \pm\ket{f(x)}\big)$. Then we have 
\begin{align*}
 \ket{0}&=\frac{1}{\sqrt 2}\big(\ket{\phi_{x+}}+\ket{\phi_{x-}}\big), \\
 \ket{f(x)}&=\frac{1}{\sqrt 2} \big(\ket{\phi_{x+}}-\ket{\phi_{x-}}\big).
\end{align*}
The intuition behind our algorithm is that the vector $\ket{\phi_{x+}}$ has a large overlap with eigenvalue-1 eigenvector of $U$ (see~(i) below), and the vector $\ket{\phi_{x-}}$ has a small overlap with eigenvectors of $U$ having small angles (see~(ii) below). So we would like to construct an algorithm in such a way that starting from the state $\ket{0}$, index of the first column of $M$, it reflects the term $\ket{\phi_{x-}}$ using the phase detection procedure and outputs $\ket{f(x)}$.

First we need to show that:
\begin{itemize}
\item[(i)] $||P_0\ket{\phi_{x+}} ||^2\geq 1-\epsilon^2,$
\item[(ii)] $\forall \Theta\geq 0:||P_\Theta\ket{\phi_{x-}} ||^2\leq \frac{\Theta^2}{4}\left(\frac{W^2}{\epsilon^2}+1 \right),$
\end{itemize}
where $P_\Theta$ is the orthogonal projection onto the span of $\{\ket{u}:U\ket{u}=e^{i\theta}\ket{u},|\theta|\leq \Theta\}$.

\paragraph{Proof of (i).} Let 
$$\ket{\psi_x}=\ket{\phi_{x+}}-\frac{\epsilon}{\sqrt{\wsize^+}}\ket{w_x},$$
where $\ket{w_x}$ is the positive witness for $x$.
Then we have
$\Pi_x\ket{\psi_x}=\ket{\psi_x}$.  Also
$\Lambda\ket{\psi_x}=\ket{\psi_x}$ since 
\begin{align*}
M\ket{\psi_x}&=\frac{A'}{\sqrt 2}\big(\ket{0}+\ket{f(x)}\big)-\frac{\epsilon}{\sqrt{\wsize^+}}A\ket{w_x}
\\ &=   \frac{\epsilon}{\sqrt{\wsize^+}}\big(\ket{z_0}-\ket{z_{f(x)}}\big)-\frac{\epsilon}{\sqrt{\wsize^+}}A\ket{w_x}
\\&=0
\end{align*}
Therefore  $U\ket{\psi_x}=\ket{\psi_x}$. We also have $|\bra{\phi_{x+}}\psi_x\rangle |^2 / \|  \ket{\psi_{x}}\|^2 \geq 1-\epsilon^2$. These give (i).

\paragraph{Proof of (ii).} Let
\begin{equation*}
\ket{\bar{\psi}_x}=\ket{\phi_{x-}}+\frac{{\sqrt{\wsize^+}}}{2\epsilon}A^\dagger \ket{\bar{w}_x},
\end{equation*}
where $\ket{\bar{w}_x}$ is the negative witness.
Then we have $\Pi_x \ket{\bar{\psi}_x}=\ket{\phi_{x-}}$.

We claim that for every $\ket{\tilde{v}}=\begin{bmatrix}
\ket{v'} \\ 
\ket{v}
\end{bmatrix} $ such that $M\ket{\tilde{v}}=0$ we have $\bra{\bar{\psi}_x}\tilde{v}\rangle=0$. This gives $\Lambda \ket{\bar{\psi}_x}=0$. To prove our claim note that if
$M\ket{\tilde{v}}=A'\ket{v'}+A\ket{v}=0,$ then
$$A\ket{v}=-A'\ket{v'}=-\frac{\sqrt 2 \epsilon}{\sqrt{\wsize^+}}\big( v'_0 \ket{z_0}-\sum_{\alpha=1}^m v'_\alpha\ket{u_\alpha}\big).$$ 
Therefore, using our assumption on negative witnesses we have
\begin{align*}
\bra{\bar{\psi}_x}\tilde{v}\big\rangle &=\frac{1}{\sqrt 2}\big(\bra{0}-\bra{f(x)}\big) \ket{v'}+\frac{\sqrt{\wsize^+}}{2\epsilon}\bra{\bar{w}_x} A\ket{v}\\ &=
v'_0-v'_{f(x)}-\big( v'_0\bra{\bar{w}_x}z_0\rangle - \sum_{\alpha=1}^m v'_\alpha \bra{\bar{w}_x}z_\alpha\rangle\big)
\\&=0
.\end{align*}

Now using Lemma~\ref{lem:eff-esp-gap} and the fact that the negative complexity is bounded by $\wsize^-$ we have 
\begin{equation*}
\|P_{\Theta} \ket{\phi_{x-}}\|^2\leq  \frac{\Theta^2}{4} \Big\| \ket{\bar{\psi}_x}\Big\|^2=
\frac{\Theta^2}{4}\left(1+\frac{\wsize^+\wsize^-}{4\epsilon^2}\right)=\frac{\Theta^2}{4}\left(1+\frac{W^2}{4\epsilon^2}\right).
\end{equation*}
This gives (ii).

\medskip
We are now ready to state the algorithm: start from state $\ket{0}$, the index of the first column of $M$; apply the phase detection circuit of Theorem~\ref{thm:phase-det} for $U=(2\Pi_x-I)(2\Lambda-I)$ with parameters $\Theta=\frac{\epsilon^2}{W}$ and $\delta=\epsilon$; measure in the computational basis and output the result.

To complete the proof we need to show that the error of this algorithm is bounded. To this end, we show that if $W>\epsilon$ then
\begin{equation*}
\| R(U)\ket{0}\ket{0^b}-\ket{f(x)}\ket{0^b}\|<4\epsilon.
\end{equation*}
Let $\bar{P_\theta}=I-P_\theta$. Then by  Theorem~\ref{thm:phase-det} and items (i) and (ii) above
we have 
\begin{align*}
& \Big\| R(U)\ket{0}\ket{0^b}-\ket{f(x)}\ket{0^b}\Big\|\\
 &=\frac{1}{\sqrt{2}} \Big\| R(U_x)(\ket{\phi_{x+}}+\ket{\phi_{x-}})\ket{0^b}-(\ket{\phi_{x+}}-\ket{\phi_{x-}})\ket{0^b}\Big\| \\ 
 & \leq  \frac{1}{\sqrt{2}}\Big \| (R(U_x)-1)\ket{\phi_{x+}}\ket{0^b}\Big\|+\frac{1}{\sqrt{2}}\Big\|(R(U_x)+1)\ket{\phi_{x-}}\ket{0^b}\Big\| \\
 & \leq 
\frac{1}{\sqrt{2}} \Big\| (R(U_x)-1)\bar{P}_0\ket{\phi_{x+}}\ket{0^b}\Big\|+\frac{1}{\sqrt{2}}\Big\| (R(U_x)+1)\bar{P}_\Theta\ket{\phi_{x-}}\ket{0^b}\Big\|+\sqrt{2}\big\|P_\Theta\ket{\phi_{x-}}\big\| \\ 
& \leq
\sqrt{2}\big\|\bar{P}_0\ket{\phi_{x+}}\big\|+ \frac{1}{\sqrt{2}} \Big\| (R(U_x)+1)\bar{P}_\Theta\ket{\phi_{x-}}\ket{0^b}\Big\|+\sqrt{2}\big\|P_\Theta\ket{\phi_{x-}}\big\| \\ 
& < 
\sqrt{2}\epsilon+\frac{\delta}{\sqrt{2}}+\frac{\Theta}{\sqrt{2}}\sqrt{\frac{W^2}{4\epsilon^2}+1} \\ 
& = \sqrt 2 \epsilon+\frac{\epsilon}{\sqrt 2}+ \frac{\epsilon^2}{\sqrt 2 W}\sqrt{\frac{W^2}{4\epsilon^2}+1} \\
& \leq \sqrt 2 \epsilon+\frac{\epsilon}{\sqrt 2}+ \epsilon \\
& < 4\epsilon.
\end{align*}
We are done.

\section{Span program for non-binary trees}\label{app:nonbinary}

In this appendix we show that given a non-binary decision tree $\cT$ of depth $T$ together with a coloring of its edges,  we can convert it to a non-binary span program (NBSP) with complexity $O(\sqrt{GT})$, where $G$ is the maximum number of red edges from the root to leaves of $\cT$.

We start by defining a non-binary decision tree for functions with non-binary input from~\cite{BT20}.
\begin{definition}[Generalized decision tree and G-coloring]
A generalized decision tree $\cT$ is a rooted directed tree such that  each internal vertex $v$ (including the root) of $\cT$ corresponds to a query index $1\leq J(v)\leq n$. Outgoing edges of $v$ are labeled by subsets of $[\ell]$ that form a partition of $[\ell]$. 
Leaves of $\cT$ are labeled with elements  of $[m]$. 
We say that $\cT$ decides a function $f: D_f\to [m]$ with $D_f\subseteq [\ell]^n$ if for every $x\in D_f$, by starting from the root of $\cT$ and following edges labeled by $Q_v(x_{J(v)})$ we reach a leaf with label $m=f(x)$.
As in the case of binary functions, a G-coloring of a generalized decision tree $\cT$ is a coloring of its edges by two colors black and red, in such a way that any vertex of $\cT$ has exactly one outgoing edge with black color.
\label{def:gen-dec-tree}
\end{definition}
We emphasize that in a generalized decision tree, some query outcomes may be grouped into a single edge if decisions made by the algorithm after this query are independent of which particular query outcome within the group occurred. Therefore labels of edges are subsets of $[\ell]$. See~\cite{BT20} for more details on this. 

We now construct a NBSP based on a generalized decision tree $\cT$. 
We note that, as in the binary case, instead of $f$ it suffices to construct an NBSP associated to the function $\tilde f$ that outputs a leaf of $\cT$ corresponding to each $x\in D_f$. 

\begin{itemize}
\item The vector space $\cV$ is $| {V}|$-dimensional with the orthonormal basis $\{\ket{v}\,:\, v \in {V}\}$, where $V$ 
is the vertex set of $ \cT$.
\item To define the vector space $H$ we first define an auxiliary space $\widehat{H}$ with the orthonormal basis 
\begin{equation*}
    \big\{ \ket{v, \black},\, \ket{v, \red}, \ket{v^\#} :\, v \in  V_\internal \big\}\cup
    \big\{\ket{v}\,:\, v \in V\big\}.
\end{equation*}
where $V_\internal$ is the set of internal vertices of $\cT$ which correspond to query indices.

Then let the input vector space $H_{j,q}\subseteq \widehat{H}$ be
\[
H_{j,q} = {\rm span} \Big\{ \ket{v, C(v, q)} - \ket{ N(v,q)} \,:\, v\in V_\internal,\, J(v) = j \Big\}, \]
where by $N(v,q)$ we mean the child of vertex $v$ in $\cT$ whose label \emph{contains} $q\in [\ell]$, and $C(v, q)$ is the color of the associated edge.
Then, let $H_j = H_{j, 0}+\cdots+ H_{j, (\ell-1)}.$

\item   $H_\free={\rm span}\left\{\ket{v^\#}: v\in {V}_{\internal}   \right\}$.

\item  $H_\forbid={\rm span}\big\{\ket{v,\red}: v\in {V}_\internal \big\}$.

\item Then let $H = H_1+\cdots +H_n + H_{\rm free}+ H_{\rm forbid}$ 
that is a subspace of $\widehat{H}$. Indeed, $H$ equals the span of the following vectors:

\begin{itemize}
   \item[-] $\ket{v, \black}-\ket{b_v} $ for $  v\in {V}_{\internal}$ where $b_v$ is the \emph{unique black child} of $v$
   \item[-] $\ket{v, \red}$ for $v\in {V}_{\internal}$
   \item[-]$\ket v$ for  $v\in V^r$
   \item[-] $\ket{v^{\#}}$ for $v\in {V}_{\internal}$
\end{itemize}
Here, by $V^r$ we mean the set of vertices of $\cT$ that have a \emph{red parent edge}.

\item The target vectors are indexed by leaves $z\in V_\leaf$ of the decision tree: 
\[\ket{t_z}=\ket{v_0}-\ket{z},\]
where $v_0$ is the root of $\cT$.

\item Define the operator $\widehat A: \widehat H \to \cV$  by 
\begin{align*}
&\widehat{A}\ket{v, c} = \sqrt{W_c}  \sum_{u \in  N(v)} \ket{u} & \forall v\in {V}_\internal, \; c \in \{\red, \black\}
\\
&\widehat{A}\ket{v} = \sqrt{W_{\black}} \ket{v} &\quad  \forall v\in {V}^b \\
&\widehat{A}\ket{v} = \sqrt{W_{\red}} \ket{v} &\quad  \forall v\in {V}^r \\
&\widehat{A}\ket{v^\#} = \ket{v} - \sum_{u \in  N(v)} \ket{u} &\quad \forall v\in {V}_\internal.
\end{align*}
Here $ N(v)$ is the set of children of $v$, $W_\black, W_\red$ are two parameters to be determined, and $V^b$ is the set of vertices of $\cT$ whose parent edge is black.
Then let $A := \widehat{A}\big|_H: H\to \cV$.
\end{itemize}

Now we show that this is a valid span program and compute its complexity. For any $x\in D_f$ let $z_{\tilde f(x)}$ be the leaf associated to $\tilde f(x)$, and let $ P(x)$ be the path from the root $z_0$ to the leaf $z_{\tilde f(x)}$. Then the positive witness $\ket{w_x}\in H(x)$ is

\begin{align*}
\ket{w_x}=&\sum_{v\in P(x)\cap V_\internal} \ket{v^\#}+\frac{1}{\sqrt{W_{C(v,x_{J(v)})}}} \Big( \ket{v, C(v, x_{J(v)})} - \ket{N(v,x_{J(v)})}\Big).
\end{align*}
Thus the positive complexity is bounded by
\begin{equation*}
   \wsize^+ \leq T +  \frac{2}{ W_\black}T+ \frac{2}{W_\red}G . 
\end{equation*}

\medskip
The negative witness equals $\ket{\bar{w}_x}=\sum_{v \in P(x)} \ket{v}\in \cV$. We note that $\bra{\bar{w}_x} t_\alpha\rangle=1-\delta_{f(x), \alpha}$ for any $\alpha$. Moreover, it is not hard to verify that $\ket{\bar{w}_x}$ is orthogonal to $AH(x)$. Thus $\ket{\bar{w}_x}$ is a valid negative witness. Next, we compute the negative complexity. 
For a vertex $v\in V_{\internal}$ let
\[H_v := {\rm span}\Big\{\frac{1}{\sqrt 2}(\ket{v,\black}-\ket{b_v)}, \ket{v, \red}\Big\}\cup \big\{  \ket{u} \,:\, u \in N(v) \cap V^r\big\} .\]
These subspaces are orthogonal to each other and to $H_{\free}$. Indeed, we have
$$H = \bigoplus_{v \in V_\internal} H_v \oplus H_{\rm free}.$$
Note that as a valid negative witness, $\ket{\bar w_x}$ is orthogonal to $AH_\free$. 
Thus, we have
\begin{align*}
\norm{ \bra{\bar{w}_x}A}^2 & =  \sum_{v \in V_\internal} \norm{\bra{\bar{w}_x}A\big|_{H_v}}^2. 
\end{align*}
Next, we note that if $v \notin P(x)$, then $\bra{\bar{w}_x}A\big|_{H_v} = 0$.
Otherwise, if $v \in P(x)$, there are two cases:
\begin{enumerate}
\item If $C(v, x_{J(v)}) = \black$, then
\begin{align*}
\norm{\bra{\bar{w}_x}A\big|_{H_v}}^2 & = \frac{1}{2} \Big|{\bra{\bar{w}_x}A \Big(\ket{v, \black} - \ket{ N(v,{x_{J(v)}})}\Big)}\Big|^2 \\& \quad +
\Big|{\bra{\bar{w}_x}A\ket{v, \red}}\Big|^2  +
\sum_{u\in N(V)\cap V^r} \Big|{\bra{\bar{w}_x}A\ket{u}}\Big|^2 \\
& = 0 + W_\red + 0.
\end{align*}
\item If $C(v, x_{J(v)}) = \red $, then
\begin{align*}
\norm{\bra{\bar{w}_x}A\big|_{H_v}}^2 & =
\frac 12 \Big|{\bra{\bar{w}_x}A\Big(\ket{v, \black} - \ket{b_v}\Big)}\Big|^2
\\ & \quad + \Big|{\bra{\bar{w}_x}A\ket{v, \red}}\Big|^2 +
\sum_{u \in N(V)\cap V^r}
\Big|{\bra{\bar{w}_x}A\ket{u}}\Big|^2 
\\ & =
\frac 12 W_\black + W_\red + W_\red.
\end{align*}
\end{enumerate}
Putting these together, we obtain
$$\wsize^-_x \le W_\red T  + (W_\black/2 + 2W_\red) G.$$ 

Now letting $W_\black=\frac{1}{G}$ and $W_\red=\frac{1}{T}$ and using $G\leq T$, we get $\wsize^+=O(GT)$ and $\wsize^-=O(1)$. This means that the total complexity is $O(\sqrt{GT})$ as desired.

\section{Proof of Theorem~\ref{thm:binaryClassical2quantumtime} for generalized non-binary decision trees}\label{app:Non-binary_Implement}

In this appendix we give the proof of Theorem~\ref{thm:binaryClassical2quantumtime} in its most general form: for $\ell>2$ and for generalized non-binary decision trees. 
The main difference of a decision tree and a non-binary decision tree is the degree of its internal vertices. We refer to Appendix~\ref{app:nonbinary} for a detailed definition of generalized non-binary decision trees.



\begin{figure}
\centering \includegraphics[scale=.7]{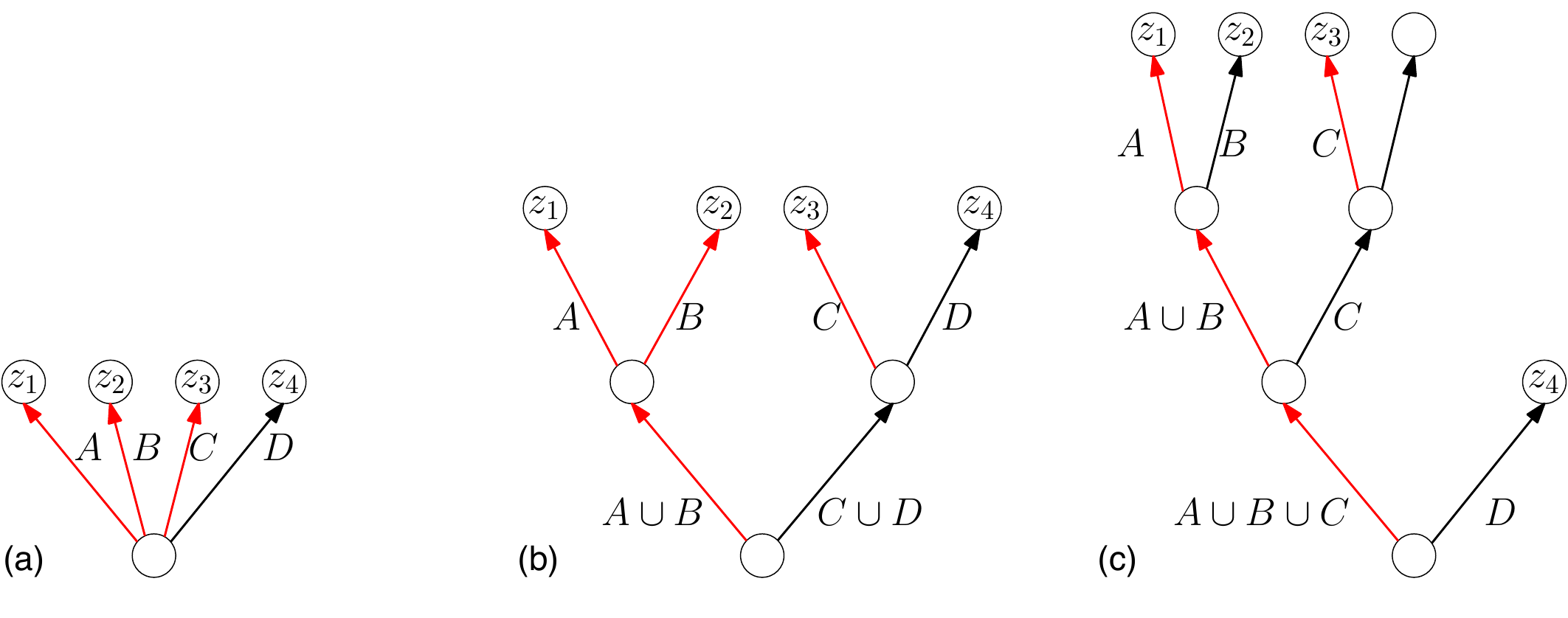}
\figcaption{  Converting a  non-binary decision tree to a binary one. A vertex of the decision tree with four children (a), can be  replaced with a depth-2 subtree (b). This approach although works increases the depth of the tree by a $\log(\ell)$ factor. In (c) we keep the black edge of (a), i.e., the edge with label $D$, untouched while replacing the other three children with a binary tree.  This process changes the depth of the tree from $T$ to $O\big(T+\log(\ell)G\big)$.
 }\label{fig:nbtobcorrect}
\end{figure}

If we use a span program with the same structure as what we defined in Theorem~\ref{thm:Classical2quantum} for a function $f:D_f\to [m]$ with $D_f\in [\ell]^n$, we get a $O(\sqrt{\ell-1})$ blow up in the complexity of the span program. An alternating solution is to use the non-binary span program that we defined in Appendix~\ref{app:nonbinary}, and to convert it to a quantum algorithm and try to implement it time-efficiently. It turns out that the construction of the kernel of the matrix $M$ for the NBSP of Appendix~\ref{app:nonbinary} is quite complicated, and we do not know how to implement the reflection through its kernel time-efficiently (even after modifying the decision tree).
Thus, our strategy to handle the case of $\ell>2$ is to turn the non-binary decision tree to a binary one at the cost of a $\sqrt{\log(\ell)}$ blow up in the query complexity. 

One may suggest that we can encode every symbol of  the input using $\log(\ell)$ binary bits and treat the function as it has a binary input (see Figure~\ref{fig:nbtobcorrect}b). Although this approach does work, it scales the query and time complexities by $\log(\ell)$.
Here, to get an improved scaling
we take a different approach. When converting a non-binary decision tree to a binary one, for any internal vertex $v$, we keep its outgoing black edge untouched, yet we replace the outgoing red edges by a binary  tree of depth $\log \ell$ (see Figure~\ref{fig:nbtobcorrect}c). This construction converts any red edge into a path of length $\log(\ell)$ in the resulting binary tree. 
This decision tree has depth $T'=T+G\log(\ell)$. Moreover, there are at most $G'=G\log(\ell)$ red edges in any path from its root to any of its leaves. So the query complexity of the resulting algorithm is $O(\sqrt{G'T'})=O(\sqrt{\log(\ell)GT+\log^2(\ell)G^2})$.
Therefore, assuming that $\log(\ell)G=O(T)$, that is usually the case in applications, we only get a $\sqrt{\log \ell}$ blow up in its query complexity. 
Then using Theorem~\ref{thm:binaryClassical2quantumtime} in the binary case, we obtain the result for arbitrary $\ell$.
Here, we should also mention that modifying the decision tree as above, imposes modifications in the subroutines $\cA_\Local$ and $\cA_\BlackPath$. These modifications are \emph{local} and prescribed as illustrated in the example of Figure~\ref{fig:nbtobcorrect}. Moreover, they increase the time complexity of the subroutines $\cA_\Local$ and $\cA_\BlackPath$ associated to the modified tree, by a constant additive term. This is because, e.g., for $\cA_{\Local}$, given any vertex of the modified tree, whose address is assumed to be given in terms of its parent vertex in the original tree, we can query $\cA_{\Local}$ of the original tree to find the children of its parent vertex in the original tree, and then compute the children in the modified tree according to Figure~\ref{fig:nbtobcorrect}. The subroutine $\cA_{\BlackPath}$ is also updated for the modified tree similarly.

\end{document}